\documentclass[12pt]{article}
\usepackage{hyperref}
\usepackage{graphicx}
\usepackage{curves}
\usepackage{amssymb,amsmath,xcolor,color,framed,epstopdf,amsthm}
\usepackage[english]{babel}
\usepackage{enumerate}
\usepackage{mathrsfs}
\usepackage{color}
\usepackage{empheq}
\usepackage{braket}
\usepackage{hyperref}
\usepackage{cancel}
\usepackage{subcaption}
 \usepackage{lipsum}
 \usepackage{mathtools}
 \usepackage{MnSymbol}
\usepackage{xcolor}
\usepackage{tikz}
\usetikzlibrary{shapes,snakes}
\usepackage{pgf}
\usepackage{pgflibraryarrows}
\usepackage{pgflibrarysnakes}
\usetikzlibrary{decorations.text}
\usepgfmodule{shapes}
\usetikzlibrary{decorations.pathmorphing}
\usetikzlibrary{decorations.markings}
\usetikzlibrary{patterns}
\usetikzlibrary{automata}
\usetikzlibrary{positioning}
\usepackage{pgfplots}
\usetikzlibrary{backgrounds}
\tikzset{partial ellipse/.style args={#1:#2:#3}{insert path={+ (#1:#3) arc (#1:#2:#3)} }}
\tikzset{->-/.style={decoration={ markings, mark=at position #1 with {\arrow{>}}},postaction={decorate}}}
\usepackage{tikz}
\usetikzlibrary{arrows.meta,decorations.markings}
\newtheorem{theorem}{Theorem}[section]
\newtheorem{lemma}[theorem]{Lemma}

\newtheorem{remark}[theorem]{Remark}

\newtheorem{RHP}{Riemann--Hilbert problem}

\newtheorem{thm}[theorem]{Theorem}

\theoremstyle{definition}
\newcommand{\R}{\mathbb{R}}
\newcommand{\Z}{\mathbb{Z}}
\newcommand{\C}{\mathbb{C}}

\newcommand{\be}{\begin{equation}}
\newcommand{\ee}{\end{equation}}
\newcommand{\bea}{\begin{eqnarray}}
\newcommand{\eea}{\end{eqnarray}}

\newtheorem{assumption}{Assumption}[section]
\newtheorem{proposition}{Proposition}[section]

\def\XXint#1#2#3{{\setbox0=\hbox{$#1{#2#3}{\int}$}
     \vcenter{\hbox{$#2#3$}}\kern-.5\wd0}}

\DeclareMathOperator*{\Res}{Res}
\def\d{{\rm d}}

\def\i{{\rm i}}

\def\Re{\operatorname{Re}}
\def\Im{\operatorname{Im}}
\def\exp{\operatorname{exp}}

\def\le{\left}
\def\ri{\right}

\tikzset{partial ellipse/.style args={#1:#2:#3}{insert path={+ (#1:#3) arc (#1:#2:#3)} }}
\tikzset{->-/.style={decoration={ markings, mark=at position #1 with {\arrow{>}}},postaction={decorate}}}
\tikzset{-<-/.style={decoration={ markings, mark=at position #1 with {\arrow{<}}},postaction={decorate}}}

\numberwithin{equation}{section}

\title{ \LARGE\bf  Long-time asymptotics of a full soliton gas for the Camassa--Holm equation}
\author{\hspace{0.6 cm}{}Dedi Yan$^{b}$, Xianguo Geng$^{a,b}$\footnote{\footnotesize
		Corresponding author.{\sl E-mail address}: xggeng@zzu.edu.cn}, Minxin Jia$^{b}$\\ 
		\leftline{\hspace{0.6 cm}{\small{\sl $^{a}$ School of Mathematics and Statistics, North China University of Water Resources }}}\\
		\leftline{\hspace{0.6 cm}{\small{\sl \quad and Electric Power, Zhengzhou, Henan 450011, People's Republic of China}}}\\
	\leftline{\hspace{0.6 cm}{\small{\sl $^{b}$ School of Mathematics and Statistics, Zhengzhou University, 100 Kexue Road, Zhengzhou, }}}\\
	\leftline{\hspace{0.6 cm}{\small{\sl \quad Henan 450001, People's Republic of China}}}}
\date{}

\begin{document}

\maketitle
\begin{abstract}
	We investigate the long-time asymptotics of a full soliton gas for the Camassa--Holm equation. The analysis starts from a pure-soliton Riemann--Hilbert (RH) problem with \(2N\) poles and two distinct types of residue conditions. We prove that, as \(N\to\infty\), this discrete RH problem converges to a limiting soliton gas RH problem whose jump matrix contains two nonzero reflection coefficients. In this sense, the limiting problem gives a full soliton gas model for the Camassa--Holm equation, in contrast to the previously studied half soliton gas models, whose jump matrices involve only one nonzero reflection coefficient.
	The limiting RH problem is analyzed by the Deift--Zhou nonlinear steepest descent method. The presence of two nonzero reflection coefficients requires two different types of triangular factorizations of the jump matrix and leads to a more delicate \(g\)-function mechanism. The main difficulty lies in the construction of suitable \(g\)-functions adapted to the Camassa--Holm phase, together with the precise control of their behavior near the distinguished point \(k=i/2\) and at infinity. Depending on the location of the spectral endpoints \(\eta_1\) and \(\eta_2\), different \(g\)-function mechanisms arise. In this paper, we focus on Case I and derive the long-time asymptotic formulas in three  elliptic-wave regions of the self-similar plane. In each region, the leading term is given by a  finite-gap elliptic function, while in the central region the first correction is of order \(\mathcal O(t^{-1/2})\) and involves parabolic cylinder functions.
\end{abstract}

\tableofcontents
\section{Introduction}
In this work, we investigate the long-time asymptotic behavior of  full soliton gas solutions for the Camassa--Holm (CH) equation
\begin{equation}\label{1.1}
	u_t-u_{txx}+2\omega u_x+3uu_x=2u_xu_{xx}+uu_{xxx},\qquad -\infty<x<\infty,\quad t>0,
\end{equation}
where \(\omega>0\). This equation was introduced as a model for the unidirectional propagation of shallow water waves over a flat bottom, with \(u(x,t)\) representing the elevation of the free surface and \(\omega\) being related to the critical shallow-water speed \cite{CH1993}. Over the past decades, the CH equation has attracted considerable attention and has been extensively studied from various perspectives~\cite{CH,Constantin0,Constantin1,Constantin7,Dai,Fan2024}.

The notion of a soliton gas, understood as an infinite statistical ensemble of interacting solitons, goes back to Zakharov's work on the Korteweg--de Vries equation \cite{Zak71}. In that setting, the main object was a rarefied gas consisting of widely separated solitons. Later, Dyachenko, Zakharov and Zakharov introduced a new model of soliton gas based on primitive potentials \cite{DZZ}. These potentials have the same spectral structure as periodic finite-gap potentials, but are in general neither periodic nor quasi-periodic. Although deterministic rather than random, they provide an effective model for integrable turbulence and soliton gases.

In the KdV case, primitive potentials are characterized by two positive H\"older-continuous functions on the spectral bands. More precisely, they are described through a vector Riemann--Hilbert(RH) problem for
\[
\Phi=(\Phi_1,\Phi_2)^{\top},
\]
which is normalized at infinity and satisfies the jump relations
\begin{align}
	\Phi_-(iz)=J(z)\Phi_+(iz),\qquad
	\Phi_-(-iz)=J(z)^{\top}\Phi_+(-iz),\qquad z\in(\eta_1,\eta_2),
\end{align}
where \(0<\eta_1<\eta_2\), and \(\Phi_\pm\) denote the boundary values on the corresponding contours. The jump matrix is given by
\begin{align}
	J(z)=\frac{1}{1+r_1(z)r_2(z)}
	\begin{pmatrix}
		1-r_1(z)r_2(z) & 2i\,r_1(z)e^{-2zx}\\
		2i\,r_2(z)e^{2zx} & 1-r_1(z)r_2(z)
	\end{pmatrix}.
\end{align}
The functions \(r_1\) and \(r_2\) play the role of reflection coefficients, and their time dependence is
\[
r_1(z;t)=r_1(z;0)e^{8z^3t},\qquad
r_2(z;t)=r_2(z;0)e^{-8z^3t}.
\]
Nabelek subsequently proved that all algebro-geometric finite-gap solutions of the KdV equation can be realized within this primitive-potential framework \cite{Nab}.

The large-scale asymptotic analysis of soliton gases has recently attracted considerable attention. Girotti et al. studied the large-space and long-time behavior of the KdV soliton gas in the case \(r_2\equiv0\), starting from the RH formulation of the \(N\)-soliton solution \cite{Girotti-1}. They also analyzed dense mKdV soliton gases and the interaction between such gases and a trial soliton \cite{Girotti-2}. Wang \emph{et al.} considered genus-two KdV soliton gases and the arbitrary-genus  defocusing nonlinear
Schr\"{o}dinger (NLS)  dark soliton gases \cite{BWYZ,Wang}. For the Camassa--Holm equation, large-space and long-time asymptotics of soliton gas solutions were obtained in \cite{GYJ}. Moreover, the authors also analyze the arbitrary odd genus mKdV soliton gas and the long-time asymptotics of focusing NLS soliton gas \cite{odd,fnls}.  Related developments include the large-\(x\) asymptotics of focusing NLS soliton gases with discrete spectrum on the imaginary axis, as well as mKdV soliton gases under nonzero boundary conditions \cite{HXF,Ling}. Recent progress has further revealed deep connections between inverse scattering theory on nontrivial backgrounds and soliton gas dynamics. Grava, Jenkins, Zhang and Zhang established an inverse scattering framework for the focusing NLS equation with elliptic background and showed its relevance to full soliton gas initial data \cite{GJZZ2}.

The purpose of the present paper is to develop a full soliton gas theory for the Camassa--Holm equation, where both reflection coefficients are nonzero. This leads to a substantially more delicate RH problem. In particular, the usual Deift--Zhou nonlinear steepest descent method must be combined with a carefully designed \(g\)-function mechanism adapted to the CH phase
\[
\theta(k)=ky-\frac{2\omega k t}{1+4k^2}.
\]
Compared with the KdV case, this phase has a more complicated analytic structure, and the reconstruction formula requires precise control of the \(g\)-function both near \(k=i/2\) and as \(k\to\infty\).

Our analysis starts from a pure-soliton RH problem for the CH equation with \(2N\) poles and two distinct types of residue conditions. We prove that, as \(N\to\infty\), this discrete RH problem converges to a continuous soliton gas RH problem; see RH problem~\ref{RHP2}. Then we give a proof of existence and uniqueness for the solution of RH problem~\ref{RHP2}; see Proposition~\ref{prop:RHP2-solvability}. By introducing two types of triangle factorization, we  apply the Deift--Zhou method \cite{DeiftZ,spin,LiuGX,GengL}, together with suitable \(g\)-function mechanism \cite{DeiftItsZhou,BL,TM,Minakov2,KM2011,KdV,Boutet,Boutet3,NLS2021}, to obtain the long-time asymptotics of the resulting soliton gas solution.

A key feature of the analysis is that the behavior depends sensitively on the location of the band endpoints \(i\eta_1\) and \(i\eta_2\). We divide the parameter space into two cases and construct different \(g\)-functions in each case. For \(0<\eta_1<\eta_2<\frac12\), define
\[
M_j=M_j(\eta_1,\eta_2):=\int_0^{\eta_1}
s^{2j}\frac{\sqrt{\eta_1^2-s^2}}
{\left(\frac14-s^2\right)^2\sqrt{\eta_2^2-s^2}}\,\mathrm ds,
\qquad j=0,1,2,
\]
and
\[
N_\pm:=\frac{2}{\frac{1}{1-4\eta_2^2}-\frac{1}{1-4\eta_1^2}\pm1}.
\]
We classify the pair \((\eta_1,\eta_2)\) as follows:
\[
\text{\rm Case I:}\qquad 1-4\frac{M_1}{M_0}> N_+,\qquad 1-4\frac{M_2}{M_1}< N_-,
\]
and
\[
\text{\rm Case II:}\qquad 1-4\frac{M_1}{M_0}< N_+.
\] The main part of this paper is devoted to Case I. The Case II can be studied in a similar way. In Case I, the \((y,t)\)-plane is divided into three elliptic-wave regions,
\[
\xi\in(-\infty,\hat{\xi}),\qquad
\xi\in(\hat{\xi},0)\cup(0,\xi_*),\qquad
\xi\in(\xi_*,+\infty),
\quad \xi=\frac{y}{\omega t}.
\]
Equivalently, in the \((x,t)\)-plane, these regions are given by
\[
\frac{x}{\omega t}\in(-\infty,\Phi(\hat{\xi})),\qquad
\frac{x}{\omega t}\in\bigl(\Phi(\hat{\xi}),\Phi(0)\bigr)\cup\bigl(\Phi(0),\Phi(\xi_*)\bigr),\qquad
\frac{x}{\omega t}\in\bigl(\Phi(\xi_*),+\infty\bigr),
\]
where $\Phi(\xi)$ is given by \eqref{Phi}.
In all three regions, the leading term is described by a modulated finite-gap elliptic function; see the asymptotic formula \eqref{uy} and \eqref{u-asymp-explicit}. The subleading term, however, depends on the region. In particular, in the central region
$
\xi\in(\hat{\xi},0)\cup(0,\xi_*),
$
the correction is of order \(\mathcal O(t^{-1/2})\), and its coefficient is expressed in terms of parabolic cylinder functions; see \eqref{u-asymp-explicit}.

The paper is organized as follows. In Section~\ref{sec2}, we formulate a pure-soliton RH problem for the CH equation with two types of residue conditions. The corresponding discrete spectrum is supported on
$
i(-\eta_2,-\eta_1)\cup i(\eta_1,\eta_2).
$
By taking the continuum limit \(N\to+\infty\), we derive a new soliton-gas RH problem for the CH equation, whose jump matrix contains two nonzero reflection coefficients on each component of the contour. In Section~\ref{sec3}, we construct the \(g\)-functions needed for the nonlinear steepest descent analysis. In particular, we control their behavior near the distinguished point \(k=i/2\) and at infinity. According to the location of the endpoints \(\eta_1\) and \(\eta_2\), the parameter space is divided into two cases, and we prove that the corresponding \(g\)-functions satisfy the required sign distributions. Section~\ref{sec4} is devoted to Case I. By introducing suitable triangular factorizations of the jump matrix, we apply the Deift--Zhou nonlinear steepest descent method to obtain the long-time asymptotic formulas for the CH soliton gas in different space-time regions.

\section{Soliton gas as limit of $N$ solitons as $N\to +\infty$}
\label{sec2}
We begin with the Lax pair for the Camassa--Holm equation \cite{CH,Monvel1},
\begin{align}\label{Lax-O}
	&-\psi_{xx}+\frac{1}{4}\psi=\lambda m\psi,\\
	&\psi_t=-\left(\frac{1}{2\lambda}+u\right)\psi_x+\frac{1}{2}u_x\psi,
\end{align}
where
\[
m=u-u_{xx}+\omega,
\]
and the positive constant $\omega$ is related to the critical shallow-water wave speed. 
The  RH problem for the pure soliton solution of the CH equation \cite{CH} is described as follows:  
\begin{RHP}\label{RHP1}
	Find  a $1\times2$ vector-valued function $M(k;y,t)$ with the following properties
	\begin{enumerate}
		\item $M(k;y,t)$ is meromorphic in $\mathbb{C}$, with simple poles at $i\kappa_{j}, iz_j $ in $i\mathbb{R}_{+}$, and at the corresponding conjugate points $-i\kappa_{j}, -iz_j $ in $i\mathbb{R}_{-}$,\ ${j=1,2,\ldots,N},\   N\in \Z_+$, $0<\kappa_{j}, z_j<\frac 12$.
		
		\item $M(k;y,t)$ satisfies the residue conditions
		\begin{equation}
			\begin{aligned}\label{residue_soliton}
				\Res\limits_{k=i\kappa_{j}}{M}(k)&=\lim_{k\to i\kappa_{j}}{M}(k)\begin{pmatrix}0&0\\i\chi_{j}e^{2i\theta(k;y,t)}&0\\\end{pmatrix},\\
				\Res\limits_{k=iz_{j}}{M}(k)&=\lim_{k\to iz_{j}}{M}(k)\begin{pmatrix}0&i\mu_{j}e^{-2i\theta(k;y,t)}\\ 0&0\\\end{pmatrix},\\
				\Res\limits_{k=-i\kappa_j}{M}(k)&
				=\lim_{k\to-i\kappa_j}{M}(k)\begin{pmatrix}0&-i\chi_je^{-2i\theta(k;y,t)}\\0&0\end{pmatrix},\\
				\Res\limits_{k=-iz_j}{M}(k)&
				=\lim_{k\to-iz_j}{M}(k)\begin{pmatrix}0&0\\ -i\mu_je^{2i\theta(k;y,t)}&0\end{pmatrix},
			\end{aligned}
		\end{equation}
		where the phase function $\theta(k)=\omega tk(\xi-\frac{2}{1+4k^2}),\ \xi=\frac{y}{\omega t}$ and $\chi_{j}, \mu_j$ are nonzero, positive and real constant.
		\item $M(k;x,t)\rightarrow \begin{pmatrix}
			1&1
		\end{pmatrix} ,\  k\rightarrow\infty$.
		\item $M(k;x,t)$ satisfies the symmetry
		\[
		\ M(-k)=M(k)\sigma_1,\ \ \sigma_1=\begin{pmatrix}
			0&1\\1&0\end{pmatrix}\ .
		\]
	\end{enumerate}
\end{RHP}
The  potential $u(x,t)$ is determined from $M$ via
\begin{align}
\label{sol}
\frac{M_{ 1}(x,t;\frac{i}{2})}{M_{ 2}(x,t;\frac{i}{2})}=\mathrm{e}^{x-y},
M_{1}(x,t;k)M_{2}(x,t;k)=\sqrt{\frac{m}{\omega}}\left(1+\frac{2i}{\omega}u(x,t)\left(k-\frac{i}{2}\right)+O\left(k-\frac{i}{2}\right)^{2}\right),
\end{align}
where $M_1(k),M_2(k)$ are the first and second entry of the vector $M(k)$, respectively.
We now are interested in the limit as $N  \to + \infty$ under the additional assumptions:
\begin{assumption}\label{assumption1}
	As $N\to\infty$, we assume that:
	\begin{enumerate}
		\item
		The poles $\{i\kappa_j\}_{j=1}^N$ and $\{iz_j\}_{j=1}^N$
		are uniformly distributed and interlaced on $i[\eta_1,\eta_2]$, where
		$0<\eta_1<\eta_2<\frac{1}{2}$, namely
		\[
		\kappa_j=\eta_1+(j-\frac14)\frac{\eta_2-\eta_1}{N},
		\qquad
		z_j=\eta_1+\left(j-\frac34\right)\frac{\eta_2-\eta_1}{N},
		\qquad j=1,\dots,N.
		\]
		
		\item
		The coefficients $\{i\chi_j\}_{j=1}^N$ and $\{i\mu_j\}_{j=1}^N$
		are purely imaginary and are discretizations of two given analytic functions:
		\begin{gather}
			i\chi_j
			=
			i\frac{\eta_2-\eta_1}{2N}r_1(i\kappa_j)
			\prod_{\substack{m=1\\ m\neq j}}^N
			\frac{\kappa_j-z_m}{\kappa_j-\kappa_m},
			\qquad j=1,\dots,N,
			\\
			i\mu_j
			=
			i\frac{\eta_2-\eta_1}{2N}\frac{\rho_1(iz_j)}{1-\rho_1(iz_j)r_1(iz_j)}
			\prod_{\substack{m=1\\ m\neq j}}^N
			\frac{z_j-\kappa_m}{z_j-z_m},
			\qquad j=1,\dots,N,
		\end{gather}
		where $r_1(k)$ and $\rho_1(k)$ are analytic functions in a neighborhood of the intervals
		$i[\eta_1,\eta_2]\cup i[-\eta_2,-\eta_1]$, satisfy
		\[
		0<r_1(k)\rho_1(k)< 1,\ k\in i(\eta_1,\eta_2),
        \]
        \[
		r_1(-k)=r_1(k),
		\quad
		\rho_1(-k)=\rho_1(k),\ k\in i(\eta_1,\eta_2),\]
		\[
		r_1(k)=\mathcal{O}(|k-i\eta_2|),\ \ \rho_1(k)=\mathcal{O}(|k-i\eta_1|),
		\]
		and is further assumed to be a real valued positive
		and non-vanishing function of $k$ for $k\in i(\eta_1,\eta_2)$.
	\end{enumerate}
	\end{assumption}

We introduce two positively oriented closed contours
\(\Gamma_{1+}\subset\mathbb{C}_{+}\) and
\(\Gamma_{1-}\subset\mathbb{C}_{-}\). The contour
\(\Gamma_{1+}\) encloses the upper spectral band
\(i[\eta_1,\eta_2]\), while \(\Gamma_{1-}\) encloses the
reflected lower band \(i[-\eta_2,-\eta_1]\). These contours are
chosen sufficiently close to the corresponding bands and do not
intersect each other, see Fig.~\ref{fig:Gamma-contours}.
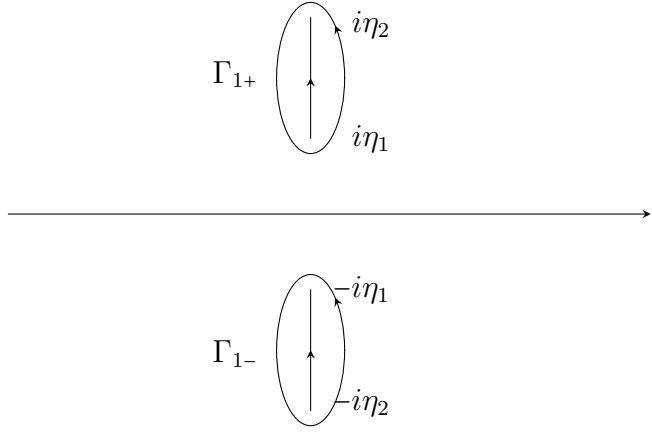
\begin{figure}[htbp]
	\centering
\begin{tikzpicture}[>=stealth,line cap=round,line join=round]
	\usetikzlibrary{arrows.meta,decorations.markings}
	
	\draw[->] (-4,0) -- (4.5,0);
	
	\draw[
	postaction={decorate},
	decoration={markings,mark=at position 0.15 with {\arrow{stealth}}}
	]
	(0.45,1.8) arc[start angle=0,end angle=360,x radius=0.45,y radius=1.0];
	
	\draw[
	postaction={decorate},
	decoration={markings,mark=at position 0.5 with {\arrow{stealth}}}
	]
	(0,1) -- (0,2.6);
	
	\node[right] at (0.4,2.5) {$i\eta_2$};
	\node[right] at (0.4,1.0) {$i\eta_1$};
	
	\draw[
	postaction={decorate},
	decoration={markings,mark=at position 0.15 with {\arrow{stealth}}}
	]
	(0.45,-1.8) arc[start angle=0,end angle=360,x radius=0.45,y radius=1.0];
	
	\draw[
	postaction={decorate},
	decoration={markings,mark=at position 0.5 with {\arrow{stealth}}}
	]
	(0,-2.6) -- (0,-1);
	\node[left] at (-0.55,1.85) {$\Gamma_{1+}$};
	\node[left] at (-0.55,-1.85) {$\Gamma_{1-}$};
	\node[right] at (0.15,-1.0) {$-i\eta_1$};
	\node[right] at (0.15,-2.5) {$-i\eta_2$};
	
\end{tikzpicture}
\caption{The closed contours $\Gamma_{1+}$ and $\Gamma_{1-}$ enclosing the spectral bands $i[\eta_1,\eta_2]$ and $i[-\eta_2,-\eta_1]$, respectively.}
\label{fig:Gamma-contours}
\end{figure}
We first remove the poles by setting
\begin{equation}\label{Zdef}
	Z(k)=M(k)
	\begin{cases}
		\begin{pmatrix}
			1&0\\
			-r_1(k)e^{2i\theta(k)}\displaystyle\prod_{j=1}^N\frac{k-iz_j}{k-i\kappa_j}&1
		\end{pmatrix}
		\begin{pmatrix}
			1&\rho_1(k)e^{-2i\theta(k)}\displaystyle\prod_{j=1}^N\frac{k-i\kappa_j}{k-iz_j}\\
			0&1
		\end{pmatrix},
		& k \text{ inside }\Gamma_{1+},
		\\[4ex]
		\begin{pmatrix}
			1&-r_1(k)e^{-2i\theta(k)}\displaystyle\prod_{j=1}^N\frac{k+iz_j}{k+i\kappa_j}\\
			0&1
		\end{pmatrix}
		\begin{pmatrix}
			1&0\\
			\rho_1(k)e^{2i\theta(k)}\displaystyle\prod_{j=1}^N\frac{k+i\kappa_j}{k+iz_j}&1
		\end{pmatrix},
		& k \text{ inside }\Gamma_{1-},
		\\[3ex]
		I,
		& \text{otherwise},
	\end{cases}
\end{equation}
where \(I\) denotes the \(2\times2\) identity matrix.

By construction, all discrete poles of \(M\) are converted into jumps
on \(\Gamma_{1+}\cup\Gamma_{1-}\). Hence the \(1\times2\) row-vector
\(Z(k;x,t)\) is analytic away from these contours and satisfies
\begin{equation}\label{Zjump-discrete}
	Z_+(k)=Z_-(k)
	\begin{cases}
		\begin{pmatrix}
			1&0\\
			-r_1(k)e^{2i\theta(k)}\displaystyle\prod_{j=1}^N\frac{k-iz_j}{k-i\kappa_j}&1
		\end{pmatrix}
		\begin{pmatrix}
			1&\rho_1(k)e^{-2i\theta(k)}\displaystyle\prod_{j=1}^N\frac{k-i\kappa_j}{k-iz_j}\\
			0&1
		\end{pmatrix},
		& k\in \Gamma_{1+},
		\\[4ex]
		\begin{pmatrix}
			1&-r_1(k)e^{-2i\theta(k)}\displaystyle\prod_{j=1}^N\frac{k+iz_j}{k+i\kappa_j}\\
			0&1
		\end{pmatrix}
		\begin{pmatrix}
			1&0\\
			\rho_1(k)e^{2i\theta(k)}\displaystyle\prod_{j=1}^N\frac{k+i\kappa_j}{k+iz_j}&1
		\end{pmatrix},
		& k\in \Gamma_{1-}.
	\end{cases}
\end{equation}
Here and below, the boundary value \(Z_+(k)\) is taken from the left
side of the oriented contour, while \(Z_-(k)\) is taken from the right
side.

To pass to the soliton-gas limit, define
\begin{equation}\label{BNdef}
	B_N(k):=\prod_{j=1}^N\frac{k-iz_j}{k-i\kappa_j},
	\qquad
	\widetilde B_N(k):=\prod_{j=1}^N\frac{k+iz_j}{k+i\kappa_j}.
\end{equation}

\begin{proposition}\label{prop:BNlimit}
	Let
	\[
	\kappa_j=\eta_1+(j-\frac14)\frac{\eta_2-\eta_1}{N},
	\qquad
	z_j=\eta_1+\left(j-\frac34\right)\frac{\eta_2-\eta_1}{N},
	\qquad j=1,\dots,N.
	\]
	Then, as \(N\to\infty\),
	\begin{gather}
		B_N(k)\longrightarrow
		\beta(k):=
		\left(\frac{k-i\eta_1}{k-i\eta_2}\right)^{1/2},
		\qquad
		k\in\mathbb{C}\setminus i[\eta_1,\eta_2],
		\label{limit1-new}
		\\
		\widetilde B_N(k)\longrightarrow
		\widetilde\beta(k):=
		\left(\frac{k+i\eta_1}{k+i\eta_2}\right)^{1/2},
		\qquad
		k\in\mathbb{C}\setminus i[-\eta_2,-\eta_1].
		\label{limit2-new}
	\end{gather}
	The convergence is uniform on compact subsets of the corresponding
	domains. The branches are fixed by the normalizations
	\[
	\beta(k)\to1,
	\qquad
	\widetilde\beta(k)\to1,
	\qquad
	k\to\infty.
	\]
\end{proposition}

\begin{proof}
	Set
	\[
	\delta_N:=\frac{\eta_2-\eta_1}{N},\qquad
	\kappa_j=\eta_1+\left(j-\frac14\right)\delta_N,\qquad
	z_j=\kappa_j-\frac{\delta_N}{2}.
	\]
	Then
	\[
	B_N(k)
	=\prod_{j=1}^N
	\left(1+\frac{i\delta_N}{2(k-i\kappa_j)}\right).
	\]
	For $k$ in a compact subset of
	$\mathbb C\setminus i[\eta_1,\eta_2]$, the expansion
	$\log(1+w)=w+\mathcal O(w^2)$ gives, uniformly on that compact set,
	\[
	\log B_N(k)
	=
	\frac12\sum_{j=1}^N
	\frac{i\delta_N}{k-i\kappa_j}
	+\mathcal O(N\delta_N^2).
	\]
	Since $N\delta_N^2\to0$, the Riemann sum
	converges uniformly on compact subsets to
	\[
	\log B_N(k)
	\longrightarrow
	\frac12\int_{\eta_1}^{\eta_2}
	\frac{i\,d\zeta}{k-i\zeta}
	=
	\frac12\log\left(
	\frac{k-i\eta_1}{k-i\eta_2}
	\right).
	\]
	Hence,
	\[
	B_N(k)\longrightarrow
	\left(\frac{k-i\eta_1}{k-i\eta_2}\right)^{1/2},
	\]
	which proves \eqref{limit1-new}. The proof of
	\eqref{limit2-new} follows in the same way.
\end{proof}

By Proposition~\ref{prop:BNlimit}, the jumps on
\(\Gamma_{1+}\cup\Gamma_{1-}\) admit a well-defined continuum limit.
Consequently, in the soliton-gas limit \(N\to\infty\),
\eqref{Zjump-discrete} becomes
\begin{equation}\label{Zjump-limit}
	Z_+(k)=Z_-(k)
	\begin{cases}
		\begin{pmatrix}
			1&0\\
			-r_1(k)e^{2i\theta(k)}\beta(k)&1
		\end{pmatrix}
		\begin{pmatrix}
			1&\rho_1(k)e^{-2i\theta(k)}\beta(k)^{-1}\\
			0&1
		\end{pmatrix},
		& k\in\Gamma_{1+},
		\\[4ex]
		\begin{pmatrix}
			1&-r_1(k)e^{-2i\theta(k)}\widetilde\beta(k)&\\
			0&1
		\end{pmatrix}
		\begin{pmatrix}
			1&0\\
			\rho_1(k)e^{2i\theta(k)}\widetilde\beta(k)^{-1}&1
		\end{pmatrix},
		& k\in\Gamma_{1-}.
	\end{cases}
\end{equation}
Next we eliminate the jumps on the  contours
\(\Gamma_{1+}\cup\Gamma_{1-}\). Define
\begin{equation}\label{Xdef}
	X(k)=Z(k)
	\begin{cases}
		\begin{pmatrix}
			1&-\rho_1(k)e^{-2i\theta(k)}\beta(k)^{-1}\\
			0&1
		\end{pmatrix}
		\begin{pmatrix}
			1&0\\
			r_1(k)e^{2i\theta(k)}\beta(k)&1
		\end{pmatrix},
		& k \text{ inside }\Gamma_{1+},
		\\[4ex]
		\begin{pmatrix}
			1&0\\
			-\rho_1(k)e^{2i\theta(k)}\widetilde\beta(k)^{-1}&1
		\end{pmatrix}
		\begin{pmatrix}
			1&r_1(k)e^{-2i\theta(k)}\widetilde\beta(k)\\
			0&1
		\end{pmatrix},
		& k \text{ inside }\Gamma_{1-},
		\\[3ex]
		I,
		& \text{otherwise}.
	\end{cases}
\end{equation}
By this transformation the jumps on \(\Gamma_{1+}\) and \(\Gamma_{1-}\)
are cancelled. Hence the only remaining jumps are supported on the
limiting spectral bands
$
i[\eta_1,\eta_2]\cup i[-\eta_2,-\eta_1].
$
With the branch choices fixed by
\[
\beta(k)=\left(\frac{k-i\eta_1}{k-i\eta_2}\right)^{1/2},
\qquad
\widetilde\beta(k)=\left(\frac{k+i\eta_1}{k+i\eta_2}\right)^{1/2},
\qquad
\beta(k),\widetilde\beta(k)\to1,\quad k\to\infty,
\]
one has
\begin{align}
	\beta_+(k)=-\beta_-(k),
	\qquad & k\in i(\eta_1,\eta_2),\\
	\widetilde\beta_+(k)=-\widetilde\beta_-(k),
	\qquad & k\in i(-\eta_2,-\eta_1),
\end{align}
where both cuts are oriented upward. We therefore arrive at the following
continuum RH problem associated with the Camassa--Holm
soliton gas.

\begin{RHP}\label{RHP2}
	Find a \(1\times2\) row-vector-valued function \(X(k;x,t)\) with the
	following properties.
	\begin{enumerate}
		\item
		\(X(k;x,t)\) is analytic for
		\[
		k\in\mathbb{C}\setminus
		\bigl(i[\eta_1,\eta_2]\cup i[-\eta_2,-\eta_1]\bigr).
		\]
		
		\item
		For \(k\in i(\eta_1,\eta_2)\cup i(-\eta_2,-\eta_1)\), the boundary
		values \(X_{\pm}(k)\) are taken from the left and right sides of the
		oriented cuts, respectively, and satisfy
		\begin{equation}\label{Xjump-factorized}
			X_+(k)=X_-(k)V_X(k).
		\end{equation}
		The jump matrix is given by the factorized form
		\begin{equation}\label{VX-factorized}
			V_X(k)=
			\begin{cases}
				\begin{pmatrix}
					1&0\\
					r_1(k)\beta_+(k)e^{2i\theta(k)}&1
				\end{pmatrix}
				\begin{pmatrix}
					1&-\frac{2\rho_1(k)}{\beta_+(k)}e^{-2i\theta(k)}\\
					0&1
				\end{pmatrix}
				\begin{pmatrix}
					1&0\\
					r_1(k)\beta_+(k)e^{2i\theta(k)}&1
				\end{pmatrix},
				& k\in i(\eta_1,\eta_2),
				\\[5ex]
				\begin{pmatrix}
					1&r_1(k)\widetilde\beta_+(k)e^{-2i\theta(k)}\\
					0&1
				\end{pmatrix}
				\begin{pmatrix}
					1&0\\
					-2\frac{\rho_1(k)}{\widetilde\beta_+(k)}e^{2i\theta(k)}&1
				\end{pmatrix}
				\begin{pmatrix}
					1&r_1(k)\widetilde\beta_+(k)e^{-2i\theta(k)}\\
					0&1
				\end{pmatrix},
				& k\in i(-\eta_2,-\eta_1).
			\end{cases}
		\end{equation}
		
		\item
		As \(k\to\infty\),
		\[
		X(k;x,t)=\begin{pmatrix}1&1\end{pmatrix}
		+\mathcal{O}(k^{-1}).
		\]
		\item At each endpoint
		$
		k_j\in\{\pm i\eta_1,\pm i\eta_2\},
		$
		the solution has at most a fourth-root singularity.
	\end{enumerate}
\end{RHP}

The symmetry of the condensed spectral data will be useful in what
follows. Since
\[
r_1(-k)=r_1(k),
\qquad
\rho_1(-k)=\rho_1(k),
\qquad
\widetilde\beta(k)=\beta(-k),
\]
and since the Camassa--Holm phase has the corresponding reflection
symmetry, the jump matrix satisfies
\[
V_X(-k)=\sigma_1V^{-1}_X(k)\sigma_1.
\]
Consequently, the solution satisfies the symmetry
$
X(-k)=X(k)\sigma_1.
$
We now rewrite the jump matrix on the upper band in a real-positive
form. For \(k\in i(\eta_1,\eta_2)\), our branch convention gives
\[
\beta_+(k)\in -i\mathbb{R}_+.
\]
Hence we introduce the real functions \( r\) and \(\rho\) by
\begin{equation}\label{hat-rho-def}
	-i r(k)=r_1(k)\beta_+(k),
	\qquad
	i\frac{\rho(k)}{1+ r(k)\rho(k)}
	=\rho_1(k)\beta_+(k)^{-1},
	\qquad k\in i(\eta_1,\eta_2).
\end{equation}
Then \( r(k)>0\) and
\[
\frac{\rho(k)}
{1+ r(k)\rho(k)}>0,
\qquad
k\in i(\eta_1,\eta_2).
\]
By assumption \ref{assumption1}, we know that $r,\rho\in L^2(i(\eta_1,\eta_2))\cap L^{\infty}(i(\eta_1,\eta_2))$.
Moreover, if the spectral densities satisfy the subcriticality condition
\[
0<r_1(k)\rho_1(k)<1,
\qquad k\in i(\eta_1,\eta_2),
\]
then \(\rho(k)>0\) on \(i(\eta_1,\eta_2)\).
For later use, we also assume $r(k)\rho(k)\neq1, k\in i(\eta_1,\eta_2)$.
In terms of \( r\) and \(\rho\), the jump matrix can be
rewritten as
\begin{equation}\label{factor1}
	V_X(k)=
	\begin{cases}
		\begin{pmatrix}
			1&0\\
			-i r(k)e^{2i\theta(k)}&1
		\end{pmatrix}
		\begin{pmatrix}
			1&-2i\dfrac{\rho(k)}
			{1+r(k)\rho(k)}e^{-2i\theta(k)}\\
			0&1
		\end{pmatrix}
		\begin{pmatrix}
			1&0\\
			-i r(k)e^{2i\theta(k)}&1
		\end{pmatrix},
		& k\in i(\eta_1,\eta_2),
		\\[5ex]
		\begin{pmatrix}
			1&i r(-k)e^{-2i\theta(k)}\\
			0&1
		\end{pmatrix}
		\begin{pmatrix}
			1&0\\
			2i\dfrac{\rho(-k)}
			{1+ r(-k)\rho(-k)}e^{2i\theta(k)}&1
		\end{pmatrix}
		\begin{pmatrix}
			1&ir(-k)e^{-2i\theta(k)}\\
			0&1
		\end{pmatrix},
		& k\in i(-\eta_2,-\eta_1).
	\end{cases}
\end{equation}
Multiplying the three factors in \eqref{factor1}, we obtain the
equivalent  representation
\begin{equation}\label{VX-expanded}
	V_X(k)=
	\begin{cases}
		\begin{pmatrix}
			\dfrac{1- r(k)\rho(k)}
			{1+ r(k)\rho(k)}
			&
			-\dfrac{2i\rho(k)}
			{1+ r(k)\rho(k)}e^{-2i\theta(k)}
			\\[2ex]
			-\dfrac{2i r(k)}
			{1+ r(k)\rho(k)}e^{2i\theta(k)}
			&
			\dfrac{1- r(k)\rho(k)}
			{1+ r(k)\rho(k)}
		\end{pmatrix},
		& k\in i(\eta_1,\eta_2),
		\\[5ex]
		\begin{pmatrix}
			\dfrac{1- r(-k)\rho(-k)}
			{1+ r(-k)\rho(-k)}
			&
			\dfrac{2i r(-k)}
			{1+ r(-k)\rho(-k)}e^{-2i\theta(k)}
			\\[2ex]
			\dfrac{2i\rho(-k)}
			{1+ r(-k)\rho(-k)}e^{2i\theta(k)}
			&
			\dfrac{1- r(-k)\rho(-k)}
			{1+ r(-k)\rho(-k)}
		\end{pmatrix},
		& k\in i(-\eta_2,-\eta_1).
	\end{cases}
\end{equation}
It is worth noting that jump matrices of this form also appear in Ref.~\cite{DZZ,GJZZ,GYW,Zhu}.

Since the above RH problem~\ref{RHP2} arises as the limit of a sequence of discrete problems, its solvability is not automatic. We give a proof of existence and uniqueness for the solution of RH problem~\ref{RHP2}.
\begin{proposition}\label{prop:RHP2-solvability}
	Since
	\[
	\rho(ip)e^{-2i\theta(ip)}>0,\qquad r(ip)e^{2i\theta(ip)}>0,\ p\in(\eta_1,\eta_2),
	\]
	then RH problem~\ref{RHP2} admits a unique solution.
\end{proposition}

\begin{proof}
	By the symmetry
	\[
	X(-k)=X(k)\sigma_1,
	\]
	it is enough to determine one scalar function \(\chi\) such that
	\[
	X(k)=\bigl(\chi(k),\chi(-k)\bigr).
	\]
	Following the argument in \cite{DZZ}, we use the Cauchy representation
	\[
	\chi(k)=1+i\int_{\eta_1}^{\eta_2}\frac{\varpi_1(s)}{k-is}\,ds
	+i\int_{\eta_1}^{\eta_2}\frac{\varpi_2(s)}{k+is}\,ds.
	\]
	By the Plemelj formula, the jump condition of RH problem~\ref{RHP2} is equivalent to the following singular integral system:
	\begin{align}
		\varpi_1(p)+a(p)\left(\int_{\eta_1}^{\eta_2}\frac{\varpi_1(q)}{p+q}\,dq+
		\operatorname{p.v.}\int_{\eta_1}^{\eta_2}\frac{\varpi_2(q)}{p-q}\,dq\right)&=a(p),\label{eq:DZZ-system-1}\\
		\varpi_2(p)+b(p)\left(\operatorname{p.v.}\int_{\eta_1}^{\eta_2}\frac{\varpi_1(q)}{p-q}\,dq+
		\int_{\eta_1}^{\eta_2}\frac{\varpi_2(q)}{p+q}\,dq\right)&=-b(p),\label{eq:DZZ-system-2}
	\end{align}
	where
	\[
	a(p)=\frac{1}{\pi}r(ip)e^{2i\theta(ip)}>0,\qquad
	b(p)=\frac{1}{\pi}\rho(ip)e^{-2i\theta(ip)}>0,\qquad p\in(\eta_1,\eta_2).
	\]
	We prove that \eqref{eq:DZZ-system-1}--\eqref{eq:DZZ-system-2} is uniquely
	solvable. Define
	\[
	(\mathcal Ch)(p):=\int_{\eta_1}^{\eta_2}\frac{h(q)}{p+q}\,dq,\qquad
	(\mathcal Hh)(p):=\operatorname{p.v.}\int_{\eta_1}^{\eta_2}\frac{h(q)}{p-q}\,dq.
	\]
	The operator \(\mathcal C\) is nonnegative. Indeed,
	\[
	\frac1{p+q}=\int_0^\infty e^{-(p+q)y}\,dy,
	\]
	and therefore
	\[
	\langle h,\mathcal Ch\rangle
	=
	\int_0^\infty
	\left|\int_{\eta_1}^{\eta_2}h(p)e^{-py}\,dp\right|^2dy
	\ge0.
	\]
	On the other hand, \(\mathcal H\) is skew-adjoint, since
	$
	\frac1{p-q}=-\frac1{q-p}.
	$
	Now set
	\[
	\varpi_1=\sqrt a\,\mathcal U_1,\qquad \varpi_2=\sqrt b\,\mathcal U_2.
	\]
	The homogeneous version of \eqref{eq:DZZ-system-1}--\eqref{eq:DZZ-system-2} then becomes
	\begin{align}
		\mathcal U_1+\sqrt a\,\mathcal C(\sqrt a\,\mathcal U_1)
		+\sqrt a\,\mathcal H(\sqrt b\,\mathcal U_2)&=0,\label{eq:U-hom-1}\\
		\mathcal U_2+\sqrt b\,\mathcal H(\sqrt a\,\mathcal U_1)
		+\sqrt b\,\mathcal C(\sqrt b\,\mathcal U_2)&=0.\label{eq:U-hom-2}
	\end{align}
	Taking the \(L^2\)-inner product of \eqref{eq:U-hom-1} with \(\mathcal U_1\), the \(L^2\)-inner product of \eqref{eq:U-hom-2} with \(\mathcal U_2\), adding the two identities and taking real parts, the terms containing \(\mathcal H\) cancel. Hence
	\[
	\|\mathcal U_1\|_2^2+\|\mathcal U_2\|_2^2
	+\langle \sqrt a\,\mathcal U_1,\mathcal C(\sqrt a\,\mathcal U_1)\rangle
	+\langle \sqrt b\,\mathcal U_2,\mathcal C(\sqrt b\,\mathcal U_2)\rangle=0.
	\]
	Every term on the left-hand side is nonnegative. Therefore
	\[
	\mathcal U_1=\mathcal U_2=0.
	\]
	Thus the homogeneous system has only the trivial solution.
	
	Moreover, the same computation gives the coercive estimate
	\[
	\|\mathcal L U\|_2\ge \|U\|_2,\qquad U=(\mathcal U_1,\mathcal U_2)^T,
	\]
	where \(\mathcal L\) denotes the operator determined by the left-hand side of
	\eqref{eq:U-hom-1}--\eqref{eq:U-hom-2}. Hence \(\operatorname{Ran}\mathcal L\)
	is closed. Applying the same estimate to the adjoint	operator \(\mathcal L^*\), we obtain
	\[
	\ker\mathcal L^*=\{0\}.
	\]
	Therefore
	\[
	(\operatorname{Ran}\mathcal L)^\perp=\ker\mathcal L^*=\{0\},
	\]
	so \(\operatorname{Ran}\mathcal L\) is dense. Since it is also closed, we have
	\[
	\operatorname{Ran}\mathcal L=L^2(\eta_1,\eta_2)\oplus L^2(\eta_1,\eta_2).
	\]
	Thus \(\mathcal L\) is invertible, and the system
	\eqref{eq:DZZ-system-1}--\eqref{eq:DZZ-system-2} has a unique solution
	\[
	(\varpi_1,\varpi_2)\in L^2(\eta_1,\eta_2)\oplus L^2(\eta_1,\eta_2).
	\]
	
	Finally, the Cauchy representation of \(\chi\) gives a solution of
	RH problem~\ref{RHP2}. If two solutions existed, their difference would solve the homogeneous problem, hence would correspond to \(\varpi_1=\varpi_2=0\), and therefore would be identically zero. This proves both existence and uniqueness.
\end{proof}
As \(N\to\infty\), the quantity \(Z\), defined by \eqref{Zjump-discrete}, converges to the solution of  \eqref{Zjump-limit}. As a consequence, the corresponding \(N\)-soliton potential \(u(x,t)\) converges to the potential reconstructed from the solution of the soliton gas RH problem ~\ref{RHP2}. Finally, the Camassa-Holm soliton gas is recovered from the vector-valued function $X(k)$ by
\begin{equation}\label{u-recover}
\frac{X_{ 1}(\frac{i}{2})}{X_{ 2}(\frac{i}{2})}=\mathrm{e}^{x-y}, \quad
X_{1}(k) X_{2}(k)=\sqrt{\frac{m}{\omega}}\left(1+\frac{2i}{\omega}u(x,t)\left(k-\frac{i}{2}\right)+\mathcal{O}\left(k-\frac{i}{2}\right)^{2}\right).
\end{equation}
\section{Construction of the \(g\)-functions in Cases I and II}\label{sec3}
For \(0<\eta_1<\eta_2<\frac12\), define
\[
M_j=M_j(\eta_1,\eta_2):=\int_0^{\eta_1}
s^{2j}\frac{\sqrt{\eta_1^2-s^2}}
{\left(\frac14-s^2\right)^2\sqrt{\eta_2^2-s^2}}\,\mathrm ds,
\qquad j=0,1,2,
\]
and
\[
N_\pm:=\frac{2}{\frac{1}{1-4\eta_2^2}-\frac{1}{1-4\eta_1^2}\pm1}.
\]
We classify the pair \((\eta_1,\eta_2)\) as follows:
\[
\text{\rm Case I:}\qquad 1-4\frac{M_1}{M_0}> N_+,\qquad 1-4\frac{M_2}{M_1}< N_-,
\]
and
\[
\text{\rm Case II:}\qquad 1-4\frac{M_1}{M_0}< N_+.
\]
\subsection{The $g$-function for Case I}
We now consider Case I. In this regime we introduce a new scalar
\(g\)-function, denoted by \(\hat g\), in order to perform the nonlinear
steepest descent analysis for the RH problem for \(X(k)\).
The function \(\hat g\) is required to satisfy the following scalar
RH conditions:
\begin{enumerate}
	\item The jump relations
	\begin{align}
		&\hat g_+(k)+\hat g_-(k)=0,
		&& k\in i(\eta_1,\eta_2)\cup i(-\eta_2,-\eta_1),\label{7.4}\\
		&\hat g_+(k)-\hat g_-(k)=\hat\Omega,\label{7.5}
		&& k\in i[-\eta_1,\eta_1].
	\end{align}
	\item The normalization at infinity
	\begin{equation}\label{7.6}
		\hat g(k)-\left(\omega k\xi-\frac{2\omega k}{1+4k^2}\right)
		=\mathcal O(k^{-1}),\qquad k\to\infty .
	\end{equation}
	\item The regularity condition at \(k=i/2\):
	\begin{equation}\label{7.7}
		\lim_{k\to i/2}
		\left[
		\hat g(k)-\left(\omega k\xi-\frac{2\omega k}{1+4k^2}\right)
		\right]
		\quad\text{exists and is finite}.
	\end{equation}
\end{enumerate}

We construct \(\hat g\) by means of two auxiliary functions. The first one is
defined by
\begin{equation}\label{g1-def}
	g_1(k;\hat\xi)
	=
	\int_{i\eta_2}^{k}
	\frac{
		\omega\hat\xi(\zeta^2+\mu_0^2)(\zeta^2-k_1^2)
		\sqrt{\zeta^2+\eta_1^2}
	}{
		\left(\zeta^2+\frac14\right)^2
		\sqrt{\zeta^2+\eta_2^2}
	}
	\,d\zeta ,
\end{equation}
As explained in section 5.3 of \cite{GYJ}, for Case I, the parameters $\hat{\xi}$, $\mu_0$, and $k_1$ are the unique solution of the following system of equations.
 \begin{subequations}\label{xtnew1}
 	\begin{align}
 		&\frac{-\hat{\xi}\sqrt{1-4\eta_1^{2}}\left(4k_{1}^{2}+1\right)(1-4\mu_{0}^{2})}{16\sqrt{1-4\eta_2^{2}}}=\frac{1}{4},\label{t1} \\
 		&1-4\mu_{0}^{2}=\frac{2(1-4\eta_2^{2})(1-4\eta_1^{2})(1+4k_{1}^{2})}{\left\{4\eta_2^{2}+4k_{1}^{2}+(1-4\eta_2^{2})4k_{1}^{2}\right\}(1-4\eta_1^{2})-(1-4\eta_2^{2})(1+4k_{1}^{2})}, \label{t2}\\
 		&\int\limits_{0}^{i\eta_1}\frac{(k^{2}-k_{1}^{2})(k^{2}+\mu_{0}^{2})\sqrt{k^{2}+\eta_1^{2}}}{\left(k^{2}+\frac{1}{4}\right)^{2}\sqrt{k^{2}+\eta_2^{2}}}\mathrm{d}k=0.\label{t3}
 	\end{align}
 	\end{subequations}
Here the branches are chosen consistently with the prescribed jumps on
\(i(\eta_1,\eta_2)\cup i(-\eta_2,-\eta_1)\). Then \(g_1\) satisfies
\begin{align}
	&g_{1+}(k)+g_{1-}(k)=0,
	&& k\in i(\eta_1,\eta_2)\cup i(-\eta_2,-\eta_1), \label{g1-jump-1}\\
	&g_{1+}(k)-g_{1-}(k)=\Omega_1,
	&& k\in i[-\eta_1,\eta_1], \label{g1-jump-2}
\end{align}
where
\begin{equation}\label{Omega1-def}
	\Omega_1
	=
	2\int_{i\eta_2}^{i\eta_1}
	\frac{
		\omega\hat\xi(\zeta^2+\mu_0^2)(\zeta^2-k_1^2)
	}{
		\left(\zeta^2+\frac14\right)^2
	}
	\left(
	\frac{\sqrt{\zeta^2+\eta_1^2}}
	{\sqrt{\zeta^2+\eta_2^2}}
	\right)_+
	d\zeta .
\end{equation}
Moreover,
\begin{equation}\label{g1-infty}
	g_1(k;\hat\xi)=\omega k\hat\xi+\mathcal O(k^{-1}),
	\qquad k\to\infty ,
\end{equation}
and the singularity at \(k=i/2\) is matched in the sense that
\begin{equation}\label{g1-half}
	\lim_{k\to i/2}
	\left[
	g_1(k;\hat\xi)
	-\left(\omega k\hat\xi-\frac{2\omega k}{1+4k^2}\right)
	\right]
	\quad\text{exists and is finite}.
\end{equation}

The second auxiliary function is given by
\begin{equation}\label{g2}
	g_2(k)
	=
	\int_{i\eta_2}^{k}
	\frac{\omega(\zeta^2+\kappa)}
	{\sqrt{(\zeta^2+\eta_1^2)(\zeta^2+\eta_2^2)}}
	\,d\zeta .
\end{equation}
It is analytic in \(\mathbb C\setminus i[-\eta_2,\eta_2]\) and satisfies
\begin{align}
	&g_{2+}(k)+g_{2-}(k)=0,
	&& k\in i(\eta_1,\eta_2)\cup i(-\eta_2,-\eta_1), \label{g2-jump-1}\\
	&g_{2+}(k)-g_{2-}(k)=\Omega_2,
	&& k\in i[-\eta_1,\eta_1], \label{g2-jump-2}\\
	&g_2(k)=\omega k+\mathcal O(k^{-1}),
	&& k\to\infty . \label{g2-infty}
\end{align}
Here
\begin{equation}\label{Omega2-def}
	\Omega_2
	=
	-\frac{\omega\pi\eta_2}{K(\hat m)},
	\qquad
	\hat m=\frac{\eta_1}{\eta_2},
	\qquad
	\kappa
	=
	\eta_2^2
	\left(
	1-\frac{E(\hat m)}{K(\hat m)}
	\right),
\end{equation}
where $K( \hat{m})=\int_0^{\frac{\pi}{2}}\frac{\d\vartheta }{\sqrt{1- \hat{m}^2\sin^2\vartheta }}$ and $E( \hat{m})=\int_0^{\frac{\pi}{2}}\sqrt{1- \hat{m}^2\sin^2\vartheta }\, \d\vartheta $ is the complete elliptic integral of the first and second kinds, respectively.

We now define
\begin{equation}\label{hatg-def}
	\hat g(k)
	=
	g_1(k;\hat\xi)+(\xi-\hat\xi)g_2(k).
\end{equation}
It follows immediately from \eqref{g1-jump-1}--\eqref{g2-infty} that
\(\hat g\) satisfies the jump conditions \eqref{7.4}--\eqref{7.5}, the
normalization \eqref{7.6}, and the regularity condition \eqref{7.7}. The
corresponding constant in the central jump is
\begin{equation}\label{hw}
	\hat\Omega
	=
	\Omega_1+(\xi-\hat\xi)\Omega_2 .
\end{equation}
In particular, the scalar RH problem
\eqref{7.4}--\eqref{7.7} determines \(\hat g\) uniquely.

Finally, near the endpoints of the cuts, the function \(\hat g\) has the
standard square-root behavior:
\begin{equation}\label{Eq1B}
	\begin{aligned}
		\hat g_+(k)-\hat g_-(k)
		&=\mathcal O\!\left(\sqrt{k\mp i\eta_2}\right),
		&& k\to \pm i\eta_2,\\
		\hat g_+(k)-\hat g_-(k)-\hat\Omega
		&=\mathcal O\!\left(\sqrt{k\mp i\eta_1}\right),
		&& k\to \pm i\eta_1.
	\end{aligned}
\end{equation}
As $\xi<\hat{\xi}$, from  \cite{GYJ}, we know the function $\hat{g}(k)$ satisfies the following inequalities:
\begin{eqnarray}
	&&\Re \le[2 i\hat{g}(k)\ri]  >0 \ \mbox{ for }k \in \mathcal{C}_{1} \setminus \{ i\eta_1, i\eta_{2} \},\label{gn11}
	\\
	&&\Re \le[2i \hat{g}(k) \ri]  <0 \ \mbox{ for }k \in \mathcal{C}_{2} \setminus \{ - i\eta_{2}, - i\eta_1 \}\label{gn12},
\end{eqnarray}
where the contours $\mathcal{C}_{1}$ and $\mathcal{C}_{2}$ shown as Fig.~\ref{openinglenses1}.
For $\xi>\hat{\xi}$, to obtain the signature table of $\Re (2i\hat{g}(k))$, we need analyze the zeros distribution of the function $\hat{g}'(k)$.
\begin{proposition}[Distribution of the zeros of $\hat g'$ for $\xi>\hat\xi$]\label{prop:zeros-hatg-xi-larger}
	Assume that \(0<\eta_1<\eta_2<\frac12\) belongs to Case I, and let \(\hat\xi\) be the transition point determined by the endpoint system \eqref{xtnew1}. Let the corresponding parameters satisfy
	\[
	0<\mu_0<\eta_1<\eta_2<\frac12,\qquad k_1>0,\qquad \hat\xi<0.
	\]
	Set
	\begin{equation}\label{eq:xi-star-hatg}
		\xi_*=\hat\xi+
		\frac{-\hat\xi(\eta_2^2-\mu_0^2)(\eta_2^2-\eta_1^2)(\eta_2^2+k_1^2)}
		{(\eta_2^2-\kappa)\left(\frac14-\eta_2^2\right)^2}.
	\end{equation}
	Then \(\xi_*>\hat\xi\). For \(\xi>\hat\xi\), the zeros of \(\hat g'(k;\xi)\) are described as follows.
	First, there is always one pair of purely imaginary zeros
	\[
	k=\pm i\beta_0(\xi),\qquad 0<\beta_0(\xi)<\eta_1.
	\]
	Second, the middle pair satisfies
	\[
	\begin{array}{ll}
		\hat\xi<\xi<\xi_*:
		& k=\pm i\beta_1(\xi),\qquad \eta_1<\beta_1(\xi)<\eta_2,\\[1mm]
		\xi=\xi_*:
		& k=\pm i\eta_2,\\[1mm]
		\xi>\xi_*:
		& k=\pm i\beta_1(\xi),\qquad \eta_2<\beta_1(\xi)<\dfrac12.
	\end{array}
	\]
	Finally, the remaining pair is distributed according to the sign of \(\xi\):
	\[
	\begin{array}{ll}
		\hat\xi<\xi<0:
		& k=\pm \gamma(\xi),\qquad \gamma(\xi)>0,\\[1mm]
		\xi=0:
		& \text{the remaining pair escapes to infinity},\\[1mm]
		\xi>0:
		& k=\pm i\beta_2(\xi),\qquad \beta_2(\xi)>\dfrac12.
	\end{array}
	\]
	In particular, for \(\hat\xi<\xi<0\), the function \(\hat g'\) has one pair of real zeros; for \(\xi>0\), all finite zeros of \(\hat g'\) lie on the imaginary axis.
\end{proposition}

\begin{proof}
	We first note that
	$
	0<\kappa<\eta_1^2.
	$
	By the fixed-endpoint construction,
	\[
	\hat g(k;\xi)=\hat g(k;\hat\xi)+(\xi-\hat\xi)g_2(k).
	\]
	Consequently,
	\[
	\hat g'(k;\xi)
	=
	\frac{\omega P_\xi(k^2)}
	{\left(k^2+\frac14\right)^2\sqrt{(k^2+\eta_1^2)(k^2+\eta_2^2)}},
	\]
	where
	\[
	P_\xi(s)
	=
	\hat\xi(s+\mu_0^2)(s+\eta_1^2)(s-k_1^2)
	+(\xi-\hat\xi)(s+\kappa)\left(s+\frac14\right)^2.
	\]
	Therefore the zeros of \(\hat g'(k;\xi)\), away from the branch points, are determined by
	$
	P_\xi(k^2)=0.
	$
	We first locate the zero closest to the origin. Since
	$
	P_\xi(-\mu_0^2)
	=
	(\xi-\hat\xi)(\kappa-\mu_0^2)\left(\frac14-\mu_0^2\right)^2
	$
	and
	$
	P_\xi(-\kappa)
	=
	\hat\xi(\mu_0^2-\kappa)(\eta_1^2-\kappa)(-\kappa-k_1^2),
	$
	these two quantities have opposite signs unless \(\mu_0^2=\kappa\), in which case \(s=-\kappa\) itself is a zero. Since both \(\mu_0^2\) and \(\kappa\) are smaller than \(\eta_1^2\), \(P_\xi\) has a zero in \((-\eta_1^2,0)\). Hence
	$
	k=\pm i\beta_0(\xi), 0<\beta_0(\xi)<\eta_1.
	$
	Next, we locate the middle zero. We have
	\[
	P_\xi(-\eta_1^2)
	=
	(\xi-\hat\xi)(\kappa-\eta_1^2)\left(\frac14-\eta_1^2\right)^2<0.
	\]
	On the other hand,
	\[
	P_\xi(-\eta_2^2)
	=
	-\hat\xi(\eta_2^2-\mu_0^2)(\eta_2^2-\eta_1^2)(\eta_2^2+k_1^2)
	-(\xi-\hat\xi)(\eta_2^2-\kappa)\left(\frac14-\eta_2^2\right)^2.
	\]
	Thus \(P_\xi(-\eta_2^2)\) is strictly decreasing as a function of \(\xi\). By the definition of \(\xi_*\) in \eqref{eq:xi-star-hatg},
	$
	P_{\xi_*}(-\eta_2^2)=0.
	$
	Therefore,
	$
	P_\xi(-\eta_2^2)>0, \hat\xi<\xi<\xi_*,
	$
	$
	P_\xi(-\eta_2^2)=0, \xi=\xi_*,
	$
	and
	$
	P_\xi(-\eta_2^2)<0, \xi>\xi_*.
	$
	If \(\hat\xi<\xi<\xi_*\), then
	$
	P_\xi(-\eta_2^2)>0, P_\xi(-\eta_1^2)<0,
	$
	so \(P_\xi\) has a zero in \((-\eta_2^2,-\eta_1^2)\). Hence,
	$
	k=\pm i\beta_1(\xi), \eta_1<\beta_1(\xi)<\eta_2.
	$
	If $\xi=\xi_*$, then \(P_{\xi_*}(-\eta_2^2)=0\), and the middle pair reaches the branch endpoints
	$
	k=\pm i\eta_2.
	$
	If \(\xi>\xi_*\), then
	$
	P_\xi(-\eta_2^2)<0,
	$
	whereas
	\[
	P_\xi\left(-\frac14\right)
	=
	\hat\xi\left(\mu_0^2-\frac14\right)
	\left(\eta_1^2-\frac14\right)
	\left(-\frac14-k_1^2\right)>0.
	\]
	Thus \(P_\xi\) has a zero in \((-\frac14,-\eta_2^2)\), which gives
	$
	k=\pm i\beta_1(\xi), \eta_2<\beta_1(\xi)<\frac12.
	$
	It remains to determine the last zero. The leading coefficient of the cubic polynomial \(P_\xi\) is
	$
	\hat\xi+(\xi-\hat\xi)=\xi.
	$
	If \(\hat\xi<\xi<0\), then
	$
	P_\xi(s)\sim \xi s^3\to -\infty, s\to+\infty.
	$
	Since
	\[
	P_\xi(0)
	=
	-\hat\xi\,\mu_0^2\eta_1^2 k_1^2
	+(\xi-\hat\xi)\frac{\kappa}{16}>0,
	\]
	there is a zero of \(P_\xi\) in \((0,+\infty)\). This gives a pair of real zeros
	$
	k=\pm\gamma(\xi), \gamma(\xi)>0.
	$
	
	If \(\xi=0\), then the cubic term of \(P_\xi\) vanishes. Hence the last zero escapes to infinity, and only the two finite pairs described above remain.
	
	If \(\xi>0\), then
	\[
	P_\xi(s)\sim \xi s^3\to -\infty,\qquad s\to-\infty.
	\]
	Since
	$
	P_\xi\left(-\frac14\right)>0,
	$
	there is a zero of \(P_\xi\) in \((-\infty,-\frac14)\). Therefore,
	$
	k=\pm i\beta_2(\xi), \beta_2(\xi)>\frac12.
	$
	For \(\xi\ne0\), the polynomial \(P_\xi\) is cubic, and the three zeros found above exhaust all its zeros. For \(\xi=0\), \(P_\xi\) is quadratic and the two finite pairs found above exhaust all finite zeros. The claimed distribution follows.
\end{proof}
By Proposition~\ref{prop:zeros-hatg-xi-larger}, we obtain the following sign properties of the phase function on the lens contours.

\begin{lemma}\label{lemma3.1}
	Assume that $\eta_1$ and $\eta_2$ belong to Case I. If
	$\xi\in(\hat{\xi},0)\cup(0,\xi_*)$, then
	\begin{align}
		&\Re\bigl(2i\hat{g}(k;\xi)\bigr)>0,
		\qquad k\in \tilde{\mathcal C}_1\setminus\{i\eta_2,i\beta_1\},\label{inn1}\\
		&\Re\bigl(2i\hat{g}(k;\xi)\bigr)<0,
		\qquad k\in \tilde{\mathcal C}_{-1}\setminus\{-i\eta_2,-i\beta_1\},\label{inn2}\\
		&\Re\bigl(2i\hat{g}(k;\xi)\bigr)<0,
		\qquad k\in \tilde{\mathcal C}_2\setminus\{i\eta_1,i\beta_1\},\label{inn3}\\
		&\Re\bigl(2i\hat{g}(k;\xi)\bigr)>0,
		\qquad k\in \tilde{\mathcal C}_{-2}\setminus\{-i\eta_1,-i\beta_1\},\label{inn4}
	\end{align}
	where the contours $\tilde{\mathcal C}_{\pm1}$ and
	$\tilde{\mathcal C}_{\pm2}$ are shown in Fig.~\ref{openinglenses2}.
	
	If $\xi>\xi_*$, then
	\begin{align}
		&\Re\bigl(2i\hat{g}(k;\xi)\bigr)<0,
		\qquad k\in {\mathcal C}_1\setminus\{i\eta_1,i\eta_2\},\label{inn5}\\
		&\Re\bigl(2i\hat{g}(k;\xi)\bigr)>0,
		\qquad k\in {\mathcal C}_2\setminus\{-i\eta_2,-i\eta_1\},\label{inn6}
	\end{align}
	where the contours ${\mathcal C}_1$ and ${\mathcal C}_2$ are shown in
	Fig.~\ref{openinglenses1}.
\end{lemma}

\begin{proof}
	We first consider the case
	$\xi\in(\hat{\xi},0)\cup(0,\xi_*)$. It is enough to prove the assertions in
	the upper half-plane, since the lower half-plane estimates follow from the
	symmetry of $\hat g$ and from the symmetric choice of the lens contours.
	
	Let $k=a+ib$, and write
	\[
	\hat g_+(k;\xi)=A_1(a,b)+iB_1(a,b),
	\qquad A_1(a,b),B_1(a,b)\in\mathbb R .
	\]
	By the definition of $\hat g$ in \eqref{hatg-def}, the boundary value
	$\hat g_+(k;\xi)$ is real on
	$i(\eta_1,\eta_2)\cup i(-\eta_2,-\eta_1)$. Hence
	\[
	B_1(0,b)=0,\qquad b\in(\eta_1,\eta_2).
	\]
	For $b\in(\beta_1,\eta_2)$, namely $k=ib\in i(\beta_1,\eta_2)$, a direct
	differentiation of \eqref{hatg-def} gives
	\begin{align}
		\left.\frac{\partial B_1}{\partial a}\right|_{a=0}
		=\Im \hat g_+'(ib;\xi)
		=\begin{cases}
			\displaystyle
			\omega\xi
			\frac{(k^2+\beta_0^2)(k^2+\beta_1^2)(k^2-\gamma^2)}
			{\left(k^2+\frac14\right)^2
				\left|\sqrt{(k^2+\eta_1^2)(k^2+\eta_2^2)}_+\right|}
			>0,
			& \xi\in(\hat{\xi},0),
			\\[3mm]
			\displaystyle
			\omega\xi
			\frac{(k^2+\beta_0^2)(k^2+\beta_1^2)(k^2+\beta_2^2)}
			{\left(k^2+\frac14\right)^2
				\left|\sqrt{(k^2+\eta_1^2)(k^2+\eta_2^2)}_+\right|}
			>0,
			& \xi\in(0,\xi_*).
		\end{cases}
	\end{align}
	Therefore, on the side of the cut where the contour
	$\tilde{\mathcal C}_1$ is placed, one has
	\[
	\Im \hat g(k;\xi)=B_1(a,b)<0,
	\qquad
	k\in\tilde{\mathcal C}_1\setminus\{i\eta_2,i\beta_1\}.
	\]
	Since
	$
	\Re\bigl(2i\hat g(k;\xi)\bigr)=-2\Im \hat g(k;\xi),
	$
	this proves \eqref{inn1}.
	
	Similarly, for $b\in(\eta_1,\beta_1)$, namely
	$k=ib\in i(\eta_1,\beta_1)$, we obtain
	\begin{align}
		\left.\frac{\partial B_1}{\partial a}\right|_{a=0}
		=\Im \hat g_+'(ib;\xi)
		=\begin{cases}
			\displaystyle
			\omega\xi
			\frac{(k^2+\beta_0^2)(k^2+\beta_1^2)(k^2-\gamma^2)}
			{\left(k^2+\frac14\right)^2
				\left|\sqrt{(k^2+\eta_1^2)(k^2+\eta_2^2)}_+\right|}
			<0,
			& \xi\in(\hat{\xi},0),
			\\[3mm]
			\displaystyle
			\omega\xi
			\frac{(k^2+\beta_0^2)(k^2+\beta_1^2)(k^2+\beta_2^2)}
			{\left(k^2+\frac14\right)^2
				\left|\sqrt{(k^2+\eta_1^2)(k^2+\eta_2^2)}_+\right|}
			<0,
			& \xi\in(0,\xi_*).
		\end{cases}
	\end{align}
	Thus
	\[
	\Im \hat g(k;\xi)>0,
	\qquad
	k\in\tilde{\mathcal C}_2\setminus\{i\eta_1,i\beta_1\},
	\]
	and hence \eqref{inn3} follows. The estimates \eqref{inn2} and \eqref{inn4}
	are obtained in the same way from the reflection symmetry.
	
	It remains to consider $\xi>\xi_*$. In this case, Proposition
	\ref{prop:zeros-hatg-xi-larger} shows that no additional zero of
	$\hat g'$ lies on the interval $i(\eta_1,\eta_2)$. Consequently, the sign of
	$\Im\hat g$ cannot change along the corresponding lens contour. Evaluating
	the sign near the cut, as above, gives
	\[
	\Im\hat g(k;\xi)>0,
	\qquad
	k\in{\mathcal C}_1\setminus\{i\eta_1,i\eta_2\}.
	\]
	Therefore,
	\[
	\Re\bigl(2i\hat g(k;\xi)\bigr)<0,
	\qquad
	k\in{\mathcal C}_1\setminus\{i\eta_1,i\eta_2\},
	\]
	which proves \eqref{inn5}. The estimate \eqref{inn6} follows again by the
	same symmetry argument. The proof is complete.
\end{proof}
\subsection{The $g$-function for Case II}
For the Case II, we first introduce a system of equations
\begin{subequations}\label{xtnew2}
	\begin{align}
		&\frac{\grave{\xi}\sqrt{1-4\eta_1^{2}}\left(4\mu_{1}^{2}-1\right)(1-4\mu_{0}^{2})}{16\sqrt{1-4\eta_2^{2}}}=\frac{1}{4},\\
		&1-4\mu_{0}^{2}=\frac{2(1-4\eta_2^{2})(1-4\eta_1^{2})(4\mu_{1}^{2}-1)}{\left[4\mu_{1}^{2}-1+(1-4\eta_2^{2})(4\mu_{1}^{2}+1)\right](1-4\eta_1^{2})-(1-4\eta_2^{2})(4\mu_{1}^{2}-1)},\\
		&\int\limits_{0}^{i\eta_1}\frac{(k^{2}+\mu_{0}^{2})(k^{2}+\mu_{1}^{2})\sqrt{k^{2}+\eta_1^{2}}}{\left(k^{2}+\frac{1}{4}\right)^{2}\sqrt{k^{2}+\eta_2^{2}}}\mathrm{d}k=0.
	\end{align}
\end{subequations}
The above system has a unique solution $\grave{\xi}, \mu_1,\mu_0$ and the parameters satisfy $\grave{\xi}>0$ and $0<\mu_0<\eta_1<\eta_2<\frac12<\mu_1$. For all $\xi\in\mathbb R$, we introduce a scalar $\grave g$-function which will be used to deform the RH problem for $X(k)$. We first construct the critical part of this function. Let
\begin{equation}\label{eq:grave-g1-def}
	\grave g_1(k;\grave\xi)
	=\int_{i\eta_2}^{k}
	\frac{\omega\grave\xi(\zeta^2+\mu_0^2)(\zeta^2+\mu_1^2)\sqrt{\zeta^2+\eta_1^2}}
	{(\zeta^2+\frac14)^2\sqrt{\zeta^2+\eta_2^2}}\,d\zeta ,
\end{equation}
where the square roots are taken with cuts along $i(\eta_1,\eta_2)\cup i(-\eta_2,-\eta_1)$, and the quotient of the square roots is normalized to be $1+\mathcal O(k^{-2})$ as $k\to\infty$. Then $\grave g_1$ is analytic in $\mathbb C\setminus(i(\eta_1,\eta_2)\cup i(-\eta_2,-\eta_1))$ and satisfies the following scalar RH conditions:
\begin{align}
	\grave g_{1+}(k)+\grave g_{1-}(k)&=0,
	&& k\in i(\eta_1,\eta_2)\cup i(-\eta_2,-\eta_1), \label{eq:grave-g1-jump-1}\\
	\grave g_{1+}(k)-\grave g_{1-}(k)&=\grave\Omega_1,
	&& k\in i[-\eta_1,\eta_1], \label{eq:grave-g1-jump-2}
\end{align}
where
\begin{equation}\label{eq:grave-Omega1}
	\grave\Omega_1
	=2\int_{i\eta_2}^{i\eta_1}
	\frac{\omega\grave\xi(\zeta^2+\mu_0^2)(\zeta^2+\mu_1^2)}
	{(\zeta^2+\frac14)^2}
	\left(\frac{\sqrt{\zeta^2+\eta_1^2}}{\sqrt{\zeta^2+\eta_2^2}}\right)_+\,d\zeta .
\end{equation}
Moreover, the normalization at infinity is
\begin{equation}\label{eq:grave-g1-infty}
	\grave g_1(k;\grave\xi)=\omega\grave\xi\, k+\mathcal O(k^{-1}),
	\qquad k\to\infty,
\end{equation}
and the pole at $k=i/2$ is matched with the original phase in the sense that
\begin{equation}\label{eq:grave-g1-i2}
	\lim_{k\to i/2}
	\left[
	\grave g_1(k;\grave\xi)
	-\left(\omega\grave\xi\, k-\frac{2\omega k}{1+4k^2}\right)
	\right]
\end{equation}
exists and is finite.

We now extend this critical $g$-function to arbitrary values of the self-similar parameter $\xi$. Define
\begin{equation}\label{eq:grave-g-def}
	\grave g(k;\xi)
	=\grave g_1(k;\grave\xi)+(\xi-\grave\xi)g_2(k),
\end{equation}
where $g_2$ is the function defined in \eqref{g2}. By the defining properties of $g_2$, the function $\grave g$ satisfies the following scalar RH problem:
\begin{align}
	\grave g_+(k)+\grave g_-(k)&=0,
	&& k\in i(\eta_1,\eta_2)\cup i(-\eta_2,-\eta_1), \label{eq:grave-g-jump-1}\\
	\grave g_+(k)-\grave g_-(k)&=\grave\Omega,
	&& k\in i[-\eta_1,\eta_1], \label{eq:grave-g-jump-2}
\end{align}
where
\begin{equation}\label{eq:grave-Omega}
	\grave\Omega=\grave\Omega_1+(\xi-\grave\xi)\Omega_2 .
\end{equation}
Furthermore, $\grave g$ has the prescribed behavior at infinity:
\begin{equation}\label{eq:grave-g-infty}
	\grave g(k;\xi)
	-\left(\omega\xi\,k-\frac{2\omega k}{1+4k^2}\right)
	=\mathcal O(k^{-1}),
	\qquad k\to\infty,
\end{equation}
and the singularity at $k=i/2$ is removable after subtracting the corresponding singular part of the phase, namely
\begin{equation}\label{eq:grave-g-i2}
	\lim_{k\to i/2}
	\left[
	\grave g(k;\xi)
	-\left(\omega\xi\,k-\frac{2\omega k}{1+4k^2}\right)
	\right]\in\mathbb C .
\end{equation}
 
\begin{proposition}[Zeros of $\grave g'$ for $\xi>\grave\xi$ in Case II]
	Assume that $0<\eta_1<\eta_2<\frac12$ belongs to Case II. Let $\grave\xi$ be the transition point determined by the endpoint system~\eqref{xtnew2}, and the corresponding parameters satisfy
	\[
	0<\mu_0<\eta_1<\eta_2<\frac12<\mu_1 .
	\]
	Set
	\begin{equation}\label{eq:xi2-grave}
		\xi_2=\grave\xi+
		\frac{\grave\xi(\eta_2^2-\mu_0^2)(\eta_2^2-\eta_1^2)(\mu_1^2-\eta_2^2)}
		{(\eta_2^2-\kappa)\left(\frac14-\eta_2^2\right)^2}.
	\end{equation}
	Then $\xi_2>\grave\xi$. For every $\xi>\grave\xi$, all zeros of $\grave g'(k;\xi)$ lie on the imaginary axis. More precisely, the zeros are given by
	\[
	k=\pm i\beta_0(\xi),\qquad k=\pm i\beta_1(\xi),\qquad k=\pm i\beta_2(\xi),
	\]
	with the following distribution:
	\[
	\begin{array}{ll}
		\grave\xi<\xi<\xi_2:
		& 0<\beta_0(\xi)<\eta_1<\beta_1(\xi)<\eta_2<\dfrac12<\beta_2(\xi)<\mu_1,\\[2mm]
		\xi=\xi_2:
		& 0<\beta_0(\xi_2)<\eta_1,\qquad \beta_1(\xi_2)=\eta_2,\qquad \dfrac12<\beta_2(\xi_2)<\mu_1,\\[2mm]
		\xi>\xi_2:
		& 0<\beta_0(\xi)<\eta_1<\eta_2<\beta_1(\xi)<\dfrac12<\beta_2(\xi)<\mu_1.
	\end{array}
	\]
	Equivalently, as $\xi$ increases past $\xi_2$, the middle pair of zeros crosses the branch endpoints $\pm i\eta_2$ and moves from the spectral bands into the gaps
	$
	i(\eta_2,\tfrac12)\cup i(-\tfrac12,-\eta_2).
	$
	No pair of zeros leaves the imaginary axis.
\end{proposition}
\begin{proof}
	By the construction of the fixed-endpoint function \(\grave g\), we have
$$
	\grave g'(k;\xi)
	=
	\frac{\omega P_\xi(k^2)}
	{\left(k^2+\frac14\right)^2
		\sqrt{(k^2+\eta_1^2)(k^2+\eta_2^2)}},
$$
	where
	\[
	P_\xi(s)
	=
	\grave\xi(s+\mu_0^2)(s+\eta_1^2)(s+\mu_1^2)
	+(\xi-\grave\xi)(s+\kappa)\left(s+\frac14\right)^2.
	\]
	Since the denominator does not create zeros away from the branch points, the zeros of
	\(\grave g'\) are determined by the equation
	$
	P_\xi(k^2)=0.
	$
	We now prove that all zeros of \(P_\xi\) are real and negative.
	First,
	$
	P_\xi(0)>0,
	$
	whereas
	\[
	P_\xi(-\eta_1^2)
	=
	(\xi-\grave\xi)(\kappa-\eta_1^2)
	\left(\frac14-\eta_1^2\right)^2<0,
	\]
	because \(\xi>\grave\xi\) and \(0<\kappa<\eta_1^2\). Therefore \(P_\xi\) has a zero in
	$
	(-\eta_1^2,0).
	$
	This gives one pair of purely imaginary zeros
	$
	k=\pm i\beta_0(\xi), 0<\beta_0(\xi)<\eta_1.
	$
	Next,
	\[
	P_\xi\left(-\frac14\right)
	=
	\grave\xi\left(\mu_0^2-\frac14\right)
	\left(\eta_1^2-\frac14\right)
	\left(\mu_1^2-\frac14\right)>0,
	\]
	since \(0<\mu_0<\eta_1<\frac12<\mu_1\). On the other hand,
	\[
	P_\xi(-\mu_1^2)
	=
	(\xi-\grave\xi)(\kappa-\mu_1^2)
	\left(\frac14-\mu_1^2\right)^2<0.
	\]
	Thus \(P_\xi\) has another zero in
	$
	(-\mu_1^2,-\tfrac14),
	$
	which gives
	$
	k=\pm i\beta_2(\xi),
	 \frac12<\beta_2(\xi)<\mu_1.
	$
	It remains to locate the middle zero. We have
	\[
	P_\xi(-\eta_2^2)
	=
	\grave\xi(\mu_0^2-\eta_2^2)(\eta_1^2-\eta_2^2)(\mu_1^2-\eta_2^2)
	+(\xi-\grave\xi)(\kappa-\eta_2^2)
	\left(\frac14-\eta_2^2\right)^2.
	\]
	Since
	\[
	(\mu_0^2-\eta_2^2)(\eta_1^2-\eta_2^2)(\mu_1^2-\eta_2^2)>0,
	\qquad
	\kappa-\eta_2^2<0,
	\]
	the quantity \(P_\xi(-\eta_2^2)\) is strictly decreasing as a function of \(\xi\). By the definition
	\[
	\xi_2=\grave\xi+
	\frac{\grave\xi(\eta_2^2-\mu_0^2)(\eta_2^2-\eta_1^2)(\mu_1^2-\eta_2^2)}
	{(\eta_2^2-\kappa)\left(\frac14-\eta_2^2\right)^2},
	\]
	we have
	$
	P_{\xi_2}(-\eta_2^2)=0.
	$
	Hence
	\[
	P_\xi(-\eta_2^2)>0,\quad \grave\xi<\xi<\xi_2,\ \ 
	P_\xi(-\eta_2^2)=0,\quad \xi=\xi_2,
	\]
	and
	$
	P_\xi(-\eta_2^2)<0, \xi>\xi_2.
	$
	
	If \(\grave\xi<\xi<\xi_2\), then
	$
	P_\xi(-\eta_2^2)>0,
	P_\xi(-\eta_1^2)<0.
	$
	Therefore \(P_\xi\) has a zero in
	$
	(-\eta_2^2,-\eta_1^2),
	$
	and so
	$
	k=\pm i\beta_1(\xi),
	\eta_1<\beta_1(\xi)<\eta_2.
	$
	If \(\xi=\xi_2\), then
	$
	P_{\xi_2}(-\eta_2^2)=0.
	$
	Thus the middle pair of zeros reaches the branch endpoints,
	$
	k=\pm i\eta_2.
	$
	Equivalently,
	$
	\beta_1(\xi_2)=\eta_2.
	$
	
	If \(\xi>\xi_2\), then
	$
	P_\xi(-\eta_2^2)<0,
	P_\xi\left(-\frac14\right)>0.
	$
	Therefore \(P_\xi\) has a zero in
	$
	(-\tfrac14,-\eta_2^2),
	$
	and hence
	$
	k=\pm i\beta_1(\xi),
	\eta_2<\beta_1(\xi)<\frac12.
	$
	
	Since \(P_\xi\) is a cubic polynomial and we have already found three zeros in mutually disjoint intervals, these are all the zeros of \(P_\xi\). Therefore all zeros of \(\grave g'(k;\xi)\) lie on the imaginary axis. The asserted classification follows immediately.
\end{proof}
An argument similar to that in Lemma~\ref{lemma3.1}  yields the following inequalities.
\begin{lemma}\label{lemma3.2} 
	Assume that $\eta_1$ and $\eta_2$ belong to Case II. If
	$\xi\in(\grave{\xi},\xi_2)$, then
	\begin{align}
		&\Re\bigl(2i\grave{g}(k;\xi)\bigr)>0,
		\qquad k\in \tilde{\mathcal C}_1\setminus\{i\eta_2,i\beta_1\},\label{ginn1}\\
		&\Re\bigl(2i\grave{g}(k;\xi)\bigr)<0,
		\qquad k\in \tilde{\mathcal C}_{-1}\setminus\{-i\eta_2,-i\beta_1\},\label{ginn2}\\
		&\Re\bigl(2i\grave{g}(k;\xi)\bigr)<0,
		\qquad k\in \tilde{\mathcal C}_2\setminus\{i\eta_1,i\beta_1\},\label{ginn3}\\
		&\Re\bigl(2i\grave{g}(k;\xi)\bigr)>0,
		\qquad k\in \tilde{\mathcal C}_{-2}\setminus\{-i\eta_1,-i\beta_1\},\label{ginn4}
	\end{align}
	where the contours $\tilde{\mathcal C}_{\pm1}$ and
	$\tilde{\mathcal C}_{\pm2}$ are shown in Fig.~\ref{openinglenses2}.
	
	If $\xi>\xi_2$, then
	\begin{align}
		&\Re\bigl(2i\grave{g}(k;\xi)\bigr)<0,
		\qquad k\in {\mathcal C}_1\setminus\{i\eta_1,i\eta_2\},\label{ginn5}\\
		&\Re\bigl(2i\grave{g}(k;\xi)\bigr)>0,
		\qquad k\in {\mathcal C}_2\setminus\{-i\eta_2,-i\eta_1\},\label{ginn6}
	\end{align}
	where the contours ${\mathcal C}_1$ and ${\mathcal C}_2$ are shown in
	Fig.~\ref{openinglenses1}.
	
	If $\xi<\grave \xi$, then
	\begin{eqnarray}
		&&\Re \left[2 i\grave{g}(k)\right]  >0 \ \mbox{ for }k \in \mathcal{C}_{1} \setminus \{ i\eta_1, i\eta_{2} \},\label{ggn11}
		\\
		&&\Re \left[2i \grave{g}(k) \right]  <0 \ \mbox{ for }k \in \mathcal{C}_{2} \setminus \{ - i\eta_{2}, - i\eta_1 \}.\label{ggn12}
	\end{eqnarray} 
\end{lemma}

\section{Behavior of the potential $u(x,t)$ as $t\rightarrow +\infty$}\label{sec4}
In this paper we mainly focus on Case I, Case II can be treated in a similar way.
Since $r(k), \rho(k)\neq 0$, one of the off-diagonal terms of the jump matrix
$V_X(k)$ is always growing. To overcome this problem, we first introduce two
new triangular factorizations for $V_X(k)$ on $i(\eta_1,\eta_2)$:
\begin{align}
	V_X(k)
	&=\begin{pmatrix}
		1&0\\
		\scriptstyle {-i\frac{r(k)\rho(k)-1}{2\rho(k)}e^{2i\theta(k)}}&1
	\end{pmatrix}
	\begin{pmatrix}
		0&\scriptstyle {-i\frac{2\rho(k)}{1+r(k)\rho(k)}e^{-2i\theta(k)}}\\
		\scriptstyle {-i\le(\frac{2\rho(k)}{1+r(k)\rho(k)}\ri)^{-1}e^{2i\theta(k)}}&0
	\end{pmatrix}
	\begin{pmatrix}
		1&0\\
		\scriptstyle {-i\frac{r(k)\rho(k)-1}{2\rho(k)}e^{2i\theta(k)}}&1
	\end{pmatrix},\label{factor11}\\
	&=\begin{pmatrix}
		1&\scriptstyle {-i\frac{r(k)\rho(k)-1}{2r(k)}e^{-2i\theta(k)}}\\
		0&1
	\end{pmatrix}
	\begin{pmatrix}
		0&\scriptstyle {-i\le(\frac{2r(k)}{1+r(k)\rho(k)}\ri)^{-1}e^{-2i\theta(k)}}\\
		\scriptstyle {-i\frac{2r(k)}{1+r(k)\rho(k)}e^{2i\theta(k)}}&0
	\end{pmatrix}
	\begin{pmatrix}
		1&\scriptstyle {-i\frac{r(k)\rho(k)-1}{2r(k)}e^{-2i\theta(k)}}\\
		0&1
	\end{pmatrix}.\label{factor12}
\end{align}
Similarly, on $i(-\eta_2,-\eta_1)$, we have
\begin{align}
	V_X(k)
	&=\begin{pmatrix}
		1&\scriptstyle {i\frac{r(-k)\rho(-k)-1}{2\rho(-k)}e^{-2i\theta(k)}}\\
		0&1
	\end{pmatrix}
	\begin{pmatrix}
		0&\scriptstyle {i\le(\frac{2\rho(-k)}{1+r(-k)\rho(-k)}\ri)^{-1}e^{-2i\theta(k)}}\\
		\scriptstyle {i\frac{2\rho(-k)}{1+r(-k)\rho(-k)}e^{2i\theta(k)}}&0
	\end{pmatrix}
	\begin{pmatrix}
		1&\scriptstyle {i\frac{r(-k)\rho(-k)-1}{2\rho(-k)}e^{-2i\theta(k)}}\\
		0&1
	\end{pmatrix},\label{factor21}\\
	&=\begin{pmatrix}
		1&0\\
		\scriptstyle {i\frac{r(-k)\rho(-k)-1}{2r(-k)}e^{2i\theta(k)}}&1
	\end{pmatrix}
	\begin{pmatrix}
		0&\scriptstyle {i\frac{2r(-k)}{1+r(-k)\rho(-k)}e^{-2i\theta(k)}}\\
		\scriptstyle {i\le(\frac{2r(-k)}{1+r(-k)\rho(-k)}\ri)^{-1}e^{2i\theta(k)}}&0
	\end{pmatrix}
	\begin{pmatrix}
		1&0\\
		\scriptstyle {i\frac{r(-k)\rho(-k)-1}{2r(-k)}e^{2i\theta(k)}}&1
	\end{pmatrix}.\label{factor22}
\end{align}
\subsection{The region $\xi>\xi_{*}$}\label{sec4.1}
By using the factorization \eqref{factor11} and \eqref{factor21}, we first open lens around $i(\eta_1,\eta_2)\cup i(-\eta_2,-\eta_1)$ by defining a new vector-valued function $Y(k)$:
\begin{align}
	Y(k)=X(k)\begin{cases}
		\begin{pmatrix} {1}& 0\\   {i  \dfrac{r(k)\rho(k)-1}{2\rho(k)} e^{2 i\theta(k)} }  & {1} \end{pmatrix},\quad k\in \text{lens left of $i(\eta_1,\eta_2)$ },\\
		\begin{pmatrix} {1}& 0\\   {-i  \dfrac{r(k)\rho(k)-1}{2\rho(k)} e^{2 i\theta(k)} }  & {1} \end{pmatrix},\quad k\in \text{lens right of $i(\eta_1,\eta_2)$ },\\
		\begin{pmatrix} 1 &    -i\dfrac{r(-k)\rho(-k)-1}{2\rho(-k)}e^{-2i\theta(k)} \\ 0 & 1\end{pmatrix}, \quad k\in \text{lens left of $i(-\eta_2,-\eta_1)$ },\\
		\begin{pmatrix} 1 &    i\dfrac{r(-k)\rho(-k)-1}{2\rho(-k)}e^{-2i\theta(k)} \\ 0 & 1\end{pmatrix}, \quad k\in \text{lens right of $i(-\eta_2,-\eta_1)$ },
	\end{cases}
\end{align}
where the lenses are shown in Fig.~\ref{openinglenses1}. Then $Y(k)$ has the jump conditions:
\begin{align}
	Y_{+}(k)=Y_{-}(k)\begin{cases}
		\begin{pmatrix} {1}& 0\\   {-i  \dfrac{r(k)\rho(k)-1}{2\rho(k)} e^{2 i\theta(k)} }  & {1} \end{pmatrix},\quad k\in \mathcal{C}_1,\\
		\begin{pmatrix} 1 &    i\dfrac{r(-k)\rho(-k)-1}{2\rho(-k)}e^{-2i\theta(k)} \\ 0 & 1\end{pmatrix}, \quad k\in\mathcal{C}_{2},\\
		\begin{pmatrix} 0 &   {-i  \dfrac{2\rho(k)}{1+r(k)\rho(k)}e^{-2 i\theta(k)} }\\  -i  \le(\dfrac{2\rho(k)}{1+r(k)\rho(k)}\ri)^{-1}e^{2 i\theta(k)}  & 0 \end{pmatrix},\quad k\in i(\eta_1,\eta_2),\\
		\begin{pmatrix} 0 &  i\le(\dfrac{2\rho(-k)}{1+r(-k)\rho(-k)}\ri)^{-1}e^{-2i\theta(k)}\\   i\dfrac{2\rho(-k)}{1+r(-k)\rho(-k)}e^{2i\theta(k)} & 0\end{pmatrix},\quad k\in i(-\eta_2,-\eta_1).
	\end{cases}
\end{align}
We introduce the following scalar function:
\begin{align}
	f(k)=\exp\left\{
	\frac{R(k)}{2\pi i}
	\left(
	\int_{i\eta_1}^{i\eta_2}
	\frac{\log \dfrac{2\rho(\zeta)}{1+r(\zeta)\rho(\zeta)}}
	{R_+(\zeta)(\zeta-k)}\,d\zeta
	-\int_{-i\eta_2}^{-i\eta_1}
	\frac{\log \dfrac{2\rho(-\zeta)}{1+r(-\zeta)\rho(-\zeta)}}
	{R_+(\zeta)(\zeta-k)}\,d\zeta
	+\int_{-i\eta_1}^{i\eta_1}
	\frac{i\Delta}{R(\zeta)(\zeta-k)}\,d\zeta
	\right)
	\right\}.
\end{align}
Here $R^{2}(k)=(k^2+\eta_1^2)(k^2+\eta_2^2)$ and the constant $\Delta$ is chosen as
\begin{align}
	\Delta
	&=
	i\left[
	\int_{i\eta_1}^{i\eta_2}
	\frac{\log \dfrac{2\rho(\zeta)}{1+r(\zeta)\rho(\zeta)}}{R_+(\zeta)}\,d\zeta
	-\int_{-i\eta_2}^{-i\eta_1}
	\frac{\log \dfrac{2\rho(-\zeta)}{1+r(-\zeta)\rho(-\zeta)}}{R_+(\zeta)}\,d\zeta
	\right]
	\left[
	\int_{-i\eta_1}^{i\eta_1}\frac{1}{R(\zeta)}\,d\zeta
	\right]^{-1} 
	\in \mathbb R .
\end{align}
Since
\[
\frac{2\rho(\zeta)}{1+r(\zeta)\rho(\zeta)}>0,
\qquad \zeta\in i(\eta_1,\eta_2),
\]
the logarithm is real-valued on $i(\eta_1,\eta_2)$. Therefore, by the symmetry of the two cuts and the above choice of branches, one obtains
$\Delta\in \mathbb R$.

It follows directly from the Plemelj formula that $f(k)$ satisfies the following scalar RH conditions:
\begin{align}
	&f_+(k)f_-(k)
	=
	\frac{2\rho(k)}{1+r(k)\rho(k)},
	\qquad
	k\in i(\eta_1,\eta_2),\\
	&f_+(k)f_-(k)
	=
	\left(
	\frac{2\rho(-k)}{1+r(-k)\rho(-k)}
	\right)^{-1},
	\qquad
	k\in i(-\eta_2,-\eta_1),\\
	&\frac{f_+(k)}{f_-(k)}
	=
	e^{i\Delta},
	\qquad
	k\in i[-\eta_1,\eta_1],\\
	&f(k)=1+\mathcal O\left(\frac{1}{k}\right),
	\qquad k\to\infty .
\end{align}

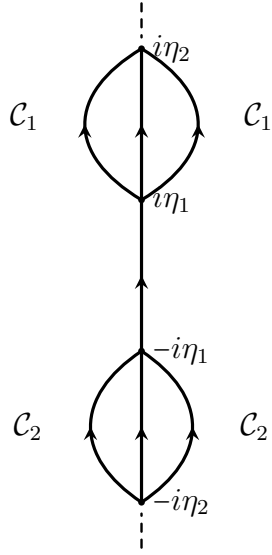
\begin{figure}[htbp]
	\centering
	\begin{tikzpicture}[
		scale=0.5,
		>=stealth,
		line cap=round,
		line join=round,
		contour/.style={
			line width=1.2pt,
			postaction={decorate},
			decoration={markings,mark=at position 0.5 with {\arrow{stealth}}}
		},
		dashedline/.style={dashed,line width=0.9pt}
		]
		
		\coordinate (A) at (0,6);   
		\coordinate (B) at (0,2);   
		\coordinate (C) at (0,-2);  
		\coordinate (D) at (0,-6);  
		
		\draw[dashedline] (0,7.2) -- (0,6.3);
		\draw[dashedline] (0,-7.2) -- (0,-6.3);
		
		\draw[contour] (D) -- (C);
		\draw[contour] (C) -- (B);
		\draw[contour] (B) -- (A);
		
		\draw[contour] (B) .. controls (-2.0,3.2) and (-2.0,4.8) .. (A);
		\draw[contour] (B) .. controls ( 2.0,3.2) and ( 2.0,4.8) .. (A);
		
		\draw[contour] (D) .. controls (-1.8,-4.8) and (-1.8,-3.2) .. (C);
		\draw[contour] (D) .. controls ( 1.8,-4.8) and ( 1.8,-3.2) .. (C);
		
		\fill (A) circle (2.5pt);
		\fill (B) circle (2.5pt);
		\fill (C) circle (2.5pt);
		\fill (D) circle (2.5pt);
		
		\node[right] at (A) {$i\eta_2$};
		\node[right] at (B) {$i\eta_1$};
		\node[right] at (C) {$-i\eta_1$};
		\node[right] at (D) {$-i\eta_2$};
		
		
		\node at (-3.1,4.2) {$\mathcal{C}_1$};
		\node at ( 3.1,4.2) {$\mathcal{C}_1$};
		\node at (-3.0,-4.0) {$\mathcal{C}_2$};
		\node at ( 3.0,-4.0) {$\mathcal{C}_2$};
		
	\end{tikzpicture}
	\caption{Opening lenses. All contours are oriented upward.}
	\label{openinglenses1}
\end{figure}
We introduce the following new vector function
\begin{equation}\label{EqD}
\begin{array}{l}
T(k)=Y(k)e^{i (t \hat{g}(k)-\theta(k))\sigma_3}f(k)^{\sigma_3},
\end{array}
\end{equation}
It is easy to see that  $T(k)$ determined by \eqref{EqD} satisfies the RH  problem:
\begin{RHP}\label{RHP3}
	Find a $1\times2$ vector-valued function ${T}(k;x,t)$ with the following properties:
	\begin{enumerate}
		\item
		${T}(k;x,t)$ is analytic for
		\[
		k\in \C\setminus \bigl(i[-\eta_2,\eta_2]\cup \mathcal{C}_1\cup \mathcal{C}_{2}\bigr).
		\]
		
		\item
		For $k\in i[-\eta_2,\eta_2]\cup \mathcal{C}_1\cup \mathcal{C}_{2}$, the boundary values ${T}_{\pm}(k)$ satisfy the jump condition
		\begin{align}
			{T}_+(k)={T}_-(k)V_{T}(k),
		\end{align}
		where
		\begin{equation}
			V_{T}(k)=
			\begin{cases}
				\begin{pmatrix}
					1&0\\[1mm]
					-i\dfrac{r(k)\rho(k)-1}{2\rho(k)}
					f^2(k)e^{2it\hat g(k)}
					&1
				\end{pmatrix},
				& k\in\mathcal{C}_1,
				\\[5mm]
				\begin{pmatrix}
					1&
					i\dfrac{r(-k)\rho(-k)-1}{2\rho(-k)}
					f^{-2}(k)e^{-2it\hat g(k)}
					\\[1mm]
					0&1
				\end{pmatrix},
				& k\in\mathcal{C}_{2},
				\\[5mm]
				\begin{pmatrix}
					0&-i\\
					-i&0
				\end{pmatrix},
				& k\in i(\eta_1,\eta_2),
				\\[5mm]
				\begin{pmatrix}
					0&i\\
					i&0
				\end{pmatrix},
				& k\in i(-\eta_2,-\eta_1),
				\\[5mm]
				\begin{pmatrix}
					e^{it\hat\Omega+i\Delta}&0\\
					0&e^{-it\hat\Omega-i\Delta}
				\end{pmatrix},
				& k\in i[-\eta_1,\eta_1].
			\end{cases}
		\end{equation}
		
		\item
		As $k\to\infty$,
		\[
		{T}(k)
		=
		\begin{pmatrix}
			1&1
		\end{pmatrix}
		+\mathcal{O}\Big(\frac1k\Big).
		\]
	\end{enumerate}
\end{RHP}
By Lemma~\ref{lemma3.1}, the off-diagonal entries of the jump matrices for $T(k)$ decay exponentially fast as $t\to+\infty$ on both the left and right lens boundaries. Consequently, one is naturally led to the following model RH problem.
\begin{RHP}\label{RHP4}
	Find a $1\times2$ vector-valued function ${T}^{\infty}(k;x,t)$ with the following properties:
	\begin{enumerate}
		\item
		${T}^{\infty}(k;x,t)$ is analytic for
		$
		k\in \C\setminus i[-\eta_2,\eta_2].
		$
		\item
		For $k\in i[-\eta_2,\eta_2]$, the boundary values ${T}^{\infty}_{\pm}(k)$ satisfy the jump condition
\begin{gather}
T^{\infty}_+(k) = T^{\infty}_-(k)
\begin{cases}
\begin{pmatrix} e^{i(t\hat\Omega+\Delta)} &0 \\ 0 & e^{-i(t\hat\Omega+\Delta)}\end{pmatrix},  & k \in i[-\eta_1,\eta_1],  \\
 \begin{pmatrix}0 & -i\\ -i & 0 \end{pmatrix}, &k \in i(\eta_1,\eta_2),  \\
  \begin{pmatrix}0 & i\\ i & 0 \end{pmatrix}, &k \in i(-\eta_2,-\eta_1),  \\
\end{cases}
\end{gather}
\item
\begin{gather}
 T^{\infty}(k)=\begin{pmatrix} 1& 1\end{pmatrix}+\mathcal{O}\le(\frac{1}{k}\ri),\quad k\to\infty.
\end{gather}
\end{enumerate}
\end{RHP}
To construct the solution of the model RH problem~\ref{RHP4}, we first introduce the elliptic
Riemann surface naturally associated with the function \(R(k)\). Let
\[
\mathfrak X=\left\{(k,\eta)\in\mathbb C^2:\;
\eta^2=R^2(k)=(k^2+\eta_1^2)(k^2+\eta_2^2)\right\}.
\]
This is a two-sheeted Riemann surface of genus one. Throughout this section,
the first sheet is fixed by the condition
$
R(k)>0, k\in\mathbb R .
$
We choose the canonical homology basis \((\mathscr A,\mathscr B)\) as shown
in Fig.~\ref{frakX} and introduce the normalized holomorphic differential
\[
\hat\omega
=\left(4\int_0^{i\eta_1}\frac{\mathrm d k}{R(k)}\right)^{-1}
\frac{\mathrm d k}{R(k)}
=\frac{\eta_2}{4iK(\hat m)}\frac{\mathrm d k}{R(k)},
\qquad
\hat m=\frac{\eta_1}{\eta_2}.
\]
With this normalization one has
\[
\oint_{\mathscr A}\hat\omega=1.
\]
The corresponding \(\mathscr B\)-period is
\[
\tau=\oint_{\mathscr B}\hat\omega
=\frac{i}{2}\frac{K\bigl(\sqrt{1-\hat m^2}\bigr)}{K(\hat m)} .
\]
We shall use the Jacobi theta function
\[
\vartheta_3(z;\tau)=\sum_{n\in\mathbb Z}
\exp\left(2\pi i n z+\pi i n^2\tau\right),
\qquad z\in\mathbb C,
\]
which is even in \(z\) and satisfies the quasi-periodicity relation
\[
\vartheta_3(z+h+\lambda\tau;\tau)
=
e^{-\pi i\lambda^2\tau-2\pi i\lambda z}
\vartheta_3(z;\tau),
\qquad h,\lambda\in\mathbb Z .
\]
Define the Abel map
\[
A(k)=\int_{i\eta_2}^{k}\hat\omega,
\qquad
k\in\mathbb C\setminus i[-\eta_2,\eta_2],
\]
where the path is taken on the first sheet. A direct computation gives
\[
A(\infty)=-\frac14,\qquad
A_+(i\eta_1)=-\frac{\tau}{2},\qquad
A_+(-i\eta_1)=-\frac12-\frac{\tau}{2},\qquad
A_+(-i\eta_2)=-\frac12 .
\]
Moreover, the boundary values of \(A\) satisfy
\[
\begin{aligned}
	&A_+(k)+A_-(k)=0,
	&& k\in\Sigma_1,\\
	&A_+(k)-A_-(k)=-\tau,
	&& k\in i[-\eta_1,\eta_1],\\
	&A_+(k)+A_-(k)=-1,
	&& k\in\Sigma_2 .
\end{aligned}
\]
We are now in a position to write down the solution of the model problem.
	The solution \(T^\infty(k)\) of the model RH problem~\ref{RHP4} is given by
	\[
	\begin{aligned}
		T^\infty(k)
		&=
		\gamma(k)
		\frac{\vartheta_3(0;2\tau)}
		{\vartheta_3\left(\dfrac{t\hat\Omega+\Delta}{2\pi};2\tau\right)}
		\begin{pmatrix}
			\displaystyle
			\frac{
				\vartheta_3\left(
				2A(k)+\dfrac{t\hat\Omega+\Delta}{2\pi}-\dfrac12;2\tau
				\right)}
			{
				\vartheta_3\left(
				2A(k)-\dfrac12;2\tau
				\right)}
			&
			\displaystyle
			\frac{
				\vartheta_3\left(
				-2A(k)+\dfrac{t\hat\Omega+\Delta}{2\pi}-\dfrac12;2\tau
				\right)}
			{
				\vartheta_3\left(
				-2A(k)-\dfrac12;2\tau
				\right)}
		\end{pmatrix},
	\end{aligned}
	\]
	where
	\[
	\gamma(k)=
	\left(\frac{k^2+\eta_1^2}{k^2+\eta_2^2}\right)^{1/4},
	\qquad
	\gamma(k)\to1,\quad k\to\infty .
	\]
\begin{figure}
\centering
\scalebox{.9}{
\begin{tikzpicture}[>=stealth]
\path (0,0) coordinate (O);

\draw (-4,-1) -- (5,-1);
\draw (-2,1) -- (7,1);
\draw (-4,-1) -- (-2,1);
\draw (5,-1) -- (7,1);
\node at (5.4,-0.1) {$\times$};
\node[above] at (5.5,-0.1) {$\infty^-$};

\draw (-4, 3) -- (5,3);
\draw (-2,5) -- (7,5);
\draw (-4,3) -- (-2,5);
\draw (5,3) -- (7,5);
\node at (5.4,3.9) {$\times$};
\node[above] at (5.5,3.9) {$\infty^+$};

\draw (-2,0) -- (.5,0);
\draw (1.5,0) -- (4,0);

\draw (-2,4) -- (.5,4);
\draw (1.5,4) -- (4,4);

\draw[dashed, black!30] (-2,0) -- (-2,4);
\draw[dashed, black!30] (.5,0) -- (.5,4);
\draw[dashed, black!30] (1.5,0) -- (1.5,4);
\draw[dashed, black!30] (4,0) -- (4,4);

\draw[fill] (-2,0) circle [radius=0.025];
\node[below ] at (-2,0) {\tiny $-i\eta_2$};
\draw[fill] (0.5,0) circle [radius=0.025];
\node[below ] at (0.5,0) {\tiny $-i\eta_1$};
\draw[fill] (1.5,0) circle [radius=0.025];
\node[below ] at (1.5,0) {\tiny $i\eta_1$};
\draw[fill] (4,0) circle [radius=0.025];
\node[below ] at (4,0) {\tiny $i\eta_2$};

\draw[fill] (-2,4) circle [radius=0.025];
\node[above ] at (-2,4) {\tiny $-i\eta_2$};
\draw[fill] (0.5,4) circle [radius=0.025];
\node[above ] at (0.5,4) {\tiny $-i\eta_1$};
\draw[fill] (1.5,4) circle [radius=0.025];
\node[above ] at (1.5,4) {\tiny $i\eta_1$};
\draw[fill] (4,4) circle [radius=0.025];
\node[above ] at (4,4) {\tiny $i\eta_2$};

\draw[->- = .25, red] (-1,4) .. controls + (70:.5cm) and + (70:.5cm) .. (2.5,4);
\draw[->- = .25, red] (2.5,0) .. controls + (-110:.5cm) and + (-110:.5cm) .. (-1,0);
\draw[red!30, dashed] (-1,0) -- (-1,4);
\draw[red!30, dashed] (2.5,0) -- (2.5,4);
\node[above, red] at (1,4.3) {\small $\mathscr{A}$};

\draw[->- = .7, blue] (1,4) .. controls + (70:.5cm) and + (70:.5cm) .. (4.5,4);
\draw[->- = .7, blue] (4.5,4) .. controls + (-110:.5cm) and + (-110:.5cm) .. (1,4);
\node[blue, above] at (3.25,4.3) {\small $\mathscr{B}$};
\end{tikzpicture}
}
\caption{The homology basis for the genus-$1$ Riemann surface $\mathfrak{X}$ .}
\label{frakX}
\end{figure}
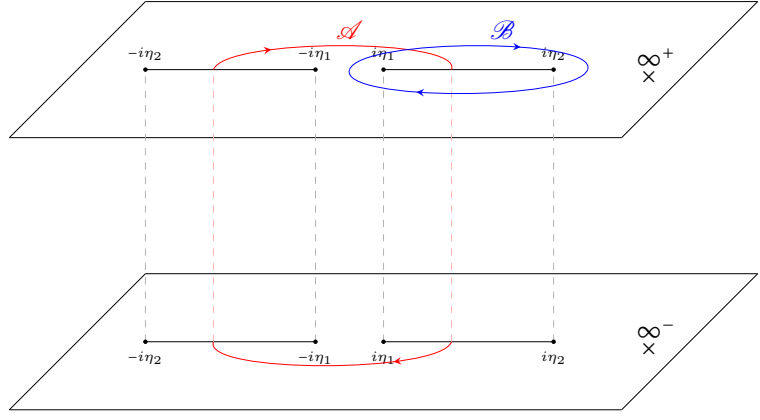
We also need to define a matrix-valued solution for later use.
Define
\begin{equation}
p(k)=i\int_{i\eta_2}^{k}\frac{\zeta^2+\kappa}
{\sqrt{(\zeta^2+\eta_1^2)(\zeta^2+\eta_2^2)}}\d \zeta,\ \ \ p(k)\to ik,\ k\to\infty,
\end{equation}
which satisfies the relations
\begin{equation}
\label{rel1}
p_+(k)+p_-(k)=0,\quad k\in i(\eta_1,\eta_2)\cup i(-\eta_2,-\eta_1),
\end{equation}
\begin{equation}
\label{rel2}
p_+(k)-p_-(k)
=
-i\frac{\pi\eta_2}{K(\hat m)},\quad k\in i(-\eta_1,\eta_1).
\end{equation}
Here, we use the same technology in \cite{Girotti-1}. Set
\begin{equation}
\label{P_infinity0}
P^{\infty}(k)=\frac{1}{2}\begin{pmatrix}
(1+\frac{p(k)}{ik})T^{\infty}_1+\dfrac{1}{ik}T^{\infty}_{1y}
&(1-\frac{p(k)}{ik})T^{\infty}_2+\dfrac{1}{ik}T^{\infty}_{2y}\\
(1-\frac{p(k)}{ik})T^{\infty}_1-\dfrac{1}{ik}T^{\infty}_{1y}&
(1+\frac{p(k)}{ik})T^{\infty}_2-\dfrac{1}{ik}T^{\infty}_{2y}
\end{pmatrix},
\end{equation}
where $T^{\infty}_1$ and $T^{\infty}_2$ denote the first and second entries of the vector-valued function $T^{\infty}$ and 
 \[
T^{\infty}_{1y}(k):=\gamma(k)\dfrac{\vartheta_3(0;2\tau)}{\vartheta_3 \le(2 A(k)  -\frac{1}{2};2\tau\ri)}
\frac{-\eta_2}{2K(\hat m)}\dfrac{\d}{\d z}\left[ \frac{\vartheta_3 \le(z+2A(k) +\frac{t \hat{\Omega}+{\Delta}}{2\pi }-\frac{1}{2};2\tau\ri)}{\vartheta_3\le(z+ \frac{t \hat{\Omega}+{\Delta}}{2\pi };2\tau\ri)}\right]\bigg|_{z=0} \ ,
\]

\[
{T}^{\infty}_{2y}(k):=\gamma(k)\dfrac{\vartheta_3(0;2\tau)}{\vartheta_3 \le(-2 A(k)  -\frac{1}{2};2\tau\ri)}
\frac{-\eta_2}{2K(\hat m)}\dfrac{\d}{\d z}\left[ \frac{\vartheta_3 \le(z-2 A(k) +\frac{t \hat{\Omega}+{\Delta}}{2\pi }-\frac{1}{2};2\tau\ri)}{\vartheta_3\le(z+ \frac{t \hat{\Omega}+{\Delta}}{2\pi };2\tau\ri)}\right]\bigg|_{z=0} \ .
\]
It follows from the construction that \(P^\infty(k)\) has the same jumps as
\(T^\infty(k)\). Moreover, its normalization at infinity is
\[
P^\infty(k)=I+\mathcal O(k^{-1}),\qquad k\to\infty .
\]
By the same argument as in \cite{Girotti-1}, one obtains
\[
\det P^\infty(k)\equiv1.
\]
Consequently, \(P^\infty(k)\) is invertible in its domain of analyticity.

We next recall the local behavior of the \(g\)-function at the branch points.
More precisely,
\[
\hat g_+(k)=\mathcal O\!\left((k\mp i\eta_2)^{1/2}\right),
\qquad k\to\pm i\eta_2,
\]
and
\[
\hat g_+(k)-\frac{\hat\Omega}{2}
=\mathcal O\!\left((k\mp i\eta_1)^{1/2}\right),
\qquad k\to\pm i\eta_1 .
\]
Therefore, in small neighborhoods of the four endpoints
\(\pm i\eta_1\) and \(\pm i\eta_2\), the local parametrices can be constructed
in the standard way by matching \(T(k)\) with the modified Bessel
model.

Combining the global parametrix, the four local parametrices, and the usual
small-norm estimates, as in \cite{Girotti-1,GYJ}, we obtain the following
expansion near \(k=i/2\):
\[
T(k)=
\left(
\mathcal E^{(0)}(y,t)
+\mathcal E^{(1)}(y,t)\left(k-\frac{i}{2}\right)
+\mathcal O\!\left(\left(k-\frac{i}{2}\right)^2\right)
\right)P^\infty(k),
\qquad k\to\frac{i}{2},
\]
where
\[
\mathcal E^{(0)}(y,t)=\begin{pmatrix}
	1&1
\end{pmatrix}+\mathcal O(t^{-1}),
\qquad
\mathcal E^{(1)}(y,t)=\mathcal O(t^{-1}).
\]

Undoing the transformations introduced above, we recover the original
RH problem unknown. Since \(X(k)=Y(k)\) for $k$ near $\frac{i}{2}$, we have
\[
X(k)
=
T(k)
\exp\{-i(t\hat g(k)-\theta(k))\sigma_3\}
f(k)^{-\sigma_3}.
\]
Hence, as \(k\to i/2\),
\[
\begin{aligned}
	X(k)
	&=
	\left(
	\mathcal E^{(0)}(y,t)
	+\mathcal E^{(1)}(y,t)\left(k-\frac{i}{2}\right)
	+\mathcal O\!\left(\left(k-\frac{i}{2}\right)^2\right)
	\right)
	P^\infty(k) \\
	&\quad \times
	\exp\{-i(t\hat g(k)-\theta(k))\sigma_3\}
	f(k)^{-\sigma_3}.
\end{aligned}
\]
Using the reconstruction formula \eqref{u-recover}, we then obtain
\[
e^{x-y}
=
\frac{T^\infty_1\!\left(y,t;\frac{i}{2}\right)}
{T^\infty_2\!\left(y,t;\frac{i}{2}\right)}
\exp\left\{-2i\left(t\hat g\!\left(\frac{i}{2}\right)
-\theta\!\left(\frac{i}{2}\right)\right)\right\}
f\!\left(\frac{i}{2}\right)^{-2}
\left(1+\mathcal O(t^{-1})\right),
\]
and
\[
\sqrt{\frac{m}{\omega}}
\left(
1+\frac{2i}{\omega}u(x,t)\left(k-\frac{i}{2}\right)
+\mathcal O\!\left(\left(k-\frac{i}{2}\right)^2\right)
\right)
=
T^\infty_1(y,t;k)T^\infty_2(y,t;k)
\left(1+\mathcal O(t^{-1})\right).
\]
	For each asymptotic sector in the \((y,t)\)-plane,
	set 
	\begin{align}\label{Phi}
	\Phi(\xi)=\xi-\frac{2i}{\omega}
	\lim_{k\to i/2}\bigl(\hat g(k;\xi)-\omega k\xi+\frac{2\omega k}{1+4k^2}\bigr).
	\end{align}
	Then the reconstruction formula at \(k=i/2\) gives
	\[
	\frac{x}{\omega t}=\Phi(\xi)+\mathcal O(t^{-1}),\qquad t\to+\infty.
	\]Consequently, if a sector in the \((y,t)\)-plane is described by
	\(\xi_a<\xi<\xi_b\), then its image in the \((x,t)\)-plane is described, up to \(\mathcal O(1)\)-shifts of the boundary curves, by
	\[
	\Phi(\xi_a)<\frac{x}{\omega t}<\Phi(\xi_b).
	\]
We therefore arrive at the following asymptotic result.
\begin{thm}\label{thm:elliptic-asymptotics}
	As \(t\to+\infty\), the full CH soliton gas \(u(x,t)\) admits the following asymptotic description. In the elliptic region
	\[
	\xi>\xi_*,
	\]
	or equivalently, in terms of the physical variable \(x\),
	\[
	\frac{x}{\omega t}>\Phi(\xi_*),
	\]
	one has
	\begin{equation}\label{uy}
		u(x,t)
		=
		\frac{4\omega(\eta_2^2-\eta_1^2)}
		{(1-4\eta_2^2)(1-4\eta_1^2)}
		-
		\frac{\eta_2\omega}{8K(\hat m)R(\frac{i}{2})}
		\frac{\hat F'\!\left(A(\frac{i}{2})\right)}
		{\hat F\!\left(A(\frac{i}{2})\right)}
		+\mathcal O(t^{-1}),
	\end{equation}
	where
	\begin{equation}\label{xy-relation}
		\begin{aligned}
			x
			&=
			y
			-2i\left(t\hat g-\theta\right)\!\left(\frac{i}{2}\right)
			+\log f^{-2}\!\left(\frac{i}{2}\right)
			+\log
			\hat F(A(\frac i2))
			+\mathcal O(t^{-1}).
		\end{aligned}
	\end{equation}
	Moreover,
	\begin{equation}\label{m-asymptotics}
		\begin{aligned}
			\sqrt{\frac{m}{\omega}}
			&=
			\sqrt{\frac{1-4\eta_1^2}{1-4\eta_2^2}}\,
			\frac{\vartheta_3^2(0;2\tau)}
			{\vartheta_3^2\!\left(\dfrac{t\hat\Omega+\Delta}{2\pi};2\tau\right)}
		\hat F(A(\frac i2))
			+\mathcal O(t^{-1}).
		\end{aligned}
	\end{equation}
	Here
	\begin{equation}\label{F-def}
		\begin{aligned}
			\hat F(U)
			&=
			\frac{
				\vartheta_3\!\left(
				2U+\dfrac{t\hat\Omega+\Delta}{2\pi}-\dfrac12;2\tau
				\right)}
			{
				\vartheta_3\!\left(
				2U-\dfrac12;2\tau
				\right)}
			\frac{
				\vartheta_3\!\left(
				-2U+\dfrac{t\hat\Omega+\Delta}{2\pi}-\dfrac12;2\tau
				\right)}
			{
				\vartheta_3\!\left(
				-2U-\dfrac12;2\tau
				\right)} ,
		\end{aligned}
	\end{equation}
	and \(F'\) denotes differentiation with respect to \(U\).
\end{thm}

\begin{remark}
	The analysis in the region \(\xi\in(-\infty,\hat{\xi})\) is essentially the same as that in the region \(\xi\in(\xi_*,+\infty)\). The only difference is that, in the present case, one has to use the alternative triangular factorizations \eqref{factor12} and \eqref{factor22}. Consequently, the resulting asymptotic formula has the same form as \eqref{uy}, with \(\Delta\) replaced by
	\begin{align}
		\hat \Delta
		&=
		-i\left[
		\int_{i\eta_1}^{i\eta_2}
		\frac{\log \dfrac{2r(\zeta)}{1+r(\zeta)\rho(\zeta)}}{R_+(\zeta)}\,d\zeta
		-\int_{-i\eta_2}^{-i\eta_1}
		\frac{\log \dfrac{2r(-\zeta)}{1+r(-\zeta)\rho(-\zeta)}}{R_+(\zeta)}\,d\zeta
		\right]
		\left[
		\int_{-i\eta_1}^{i\eta_1}\frac{1}{R(\zeta)}\,d\zeta
		\right]^{-1}
		\in \mathbb R .
	\end{align}
\end{remark}
\subsection{The Region $\xi\in (\hat{\xi},0)\cup(0,\xi_{*})$}
For $\xi\in (\hat{\xi},0)\cup(0,\xi_{*})$, Proposition~\ref{prop:zeros-hatg-xi-larger} implies that
the phase $\hat{g}(k)$ possesses precisely two stationary points,
located at $\pm i\beta_1$, where $\beta_1\in(\eta_1,\eta_2)$. We therefore open lenses along the intervals
$i(\eta_1,\beta_1)$ and $i(\beta_1,\eta_2)$, together with their symmetric
counterparts in the lower half-plane. The choice of lenses is dictated by
the sign distribution of the phase and by the factorizations
\eqref{factor11}--\eqref{factor22}. We define a new unknown $\tilde Y$ by
\begin{align}
	\tilde Y(k)=X(k)
	\begin{cases}
		\begin{pmatrix}
			1&0\\
			\displaystyle i\frac{r(k)\rho(k)-1}{2\rho(k)}e^{2i\theta(k)}&1
		\end{pmatrix},
		& k\in \text{the left lens adjacent to }i(\eta_1,\beta_1),
		\\[3mm]
		\begin{pmatrix}
			1&0\\
			\displaystyle -i\frac{r(k)\rho(k)-1}{2\rho(k)}e^{2i\theta(k)}&1
		\end{pmatrix},
		& k\in \text{the right lens adjacent to }i(\eta_1,\beta_1),
		\\[3mm]
		\begin{pmatrix}
			1&\displaystyle i\frac{r(k)\rho(k)-1}{2r(k)}e^{-2i\theta(k)}\\
			0&1
		\end{pmatrix},
		& k\in \text{the left lens adjacent to }i(\beta_1,\eta_2),
		\\[3mm]
		\begin{pmatrix}
			1&\displaystyle -i\frac{r(k)\rho(k)-1}{2r(k)}e^{-2i\theta(k)}\\
			0&1
		\end{pmatrix},
		& k\in \text{the right lens adjacent to }i(\beta_1,\eta_2),
		\\[3mm]
		\begin{pmatrix}
			1&\displaystyle -i\frac{r(-k)\rho(-k)-1}{2\rho(-k)}e^{-2i\theta(k)}\\
			0&1
		\end{pmatrix},
		& k\in \text{the left lens adjacent to }i(-\beta_1,-\eta_1),
		\\[3mm]
		\begin{pmatrix}
			1&\displaystyle i\frac{r(-k)\rho(-k)-1}{2\rho(-k)}e^{-2i\theta(k)}\\
			0&1
		\end{pmatrix},
		& k\in \text{the right lens adjacent to }i(-\beta_1,-\eta_1),
		\\[3mm]
		\begin{pmatrix}
			1&0\\
			\displaystyle -i\frac{r(-k)\rho(-k)-1}{2r(-k)}e^{2i\theta(k)}&1
		\end{pmatrix},
		& k\in \text{the left lens adjacent to }i(-\eta_2,-\beta_1),
		\\[3mm]
		\begin{pmatrix}
			1&0\\
			\displaystyle i\frac{r(-k)\rho(-k)-1}{2r(-k)}e^{2i\theta(k)}&1
		\end{pmatrix},
		& k\in \text{the right lens adjacent to }i(-\eta_2,-\beta_1).
	\end{cases}
\end{align}
The lens contours are shown in Fig.~\ref{openinglenses2}. With this definition,
the jumps on the lips of the lenses are triangular, whereas the jumps on the
original cuts are reduced to purely off-diagonal form. More precisely, $\tilde Y$
satisfies
\begin{align}
	\tilde Y_+(k)=\tilde Y_-(k)
	\begin{cases}
		\begin{pmatrix}
			1&0\\
			\displaystyle -i\frac{r(k)\rho(k)-1}{2\rho(k)}e^{2i\theta(k)}&1
		\end{pmatrix},
		& k\in \tilde{\mathcal C}_2,
		\\[3mm]
		\begin{pmatrix}
			1&\displaystyle -i\frac{r(k)\rho(k)-1}{2r(k)}e^{-2i\theta(k)}\\
			0&1
		\end{pmatrix},
		& k\in \tilde{\mathcal C}_1,
		\\[3mm]
		\begin{pmatrix}
			1&\displaystyle i\frac{r(-k)\rho(-k)-1}{2\rho(-k)}e^{-2i\theta(k)}\\
			0&1
		\end{pmatrix},
		& k\in \tilde{\mathcal C}_{-2},
		\\[3mm]
		\begin{pmatrix}
			1&0\\
			\displaystyle i\frac{r(-k)\rho(-k)-1}{2r(-k)}e^{2i\theta(k)}&1
		\end{pmatrix},
		& k\in \tilde{\mathcal C}_{-1},
		\\[3mm]
		\begin{pmatrix}
			0&\displaystyle -i\frac{2\rho(k)}{1+r(k)\rho(k)}e^{-2i\theta(k)}\\
			\displaystyle -i\left(\frac{2\rho(k)}{1+r(k)\rho(k)}\right)^{-1}e^{2i\theta(k)}&0
		\end{pmatrix},
		& k\in (i\eta_1,i\beta_1),
		\\[3mm]
		\begin{pmatrix}
			0&\displaystyle -i\left(\frac{2r(k)}{1+r(k)\rho(k)}\right)^{-1}e^{-2i\theta(k)}\\
			\displaystyle -i\frac{2r(k)}{1+r(k)\rho(k)}e^{2i\theta(k)}&0
		\end{pmatrix},
		& k\in (i\beta_1,i\eta_2),
		\\[3mm]
		\begin{pmatrix}
			0&\displaystyle i\left(\frac{2\rho(-k)}{1+r(-k)\rho(-k)}\right)^{-1}e^{-2i\theta(k)}\\
			\displaystyle i\frac{2\rho(-k)}{1+r(-k)\rho(-k)}e^{2i\theta(k)}&0
		\end{pmatrix},
		& k\in (-i\beta_1,-i\eta_1),
		\\[3mm]
		\begin{pmatrix}
			0&\displaystyle i\frac{2r(-k)}{1+r(-k)\rho(-k)}e^{-2i\theta(k)}\\
			\displaystyle i\left(\frac{2r(-k)}{1+r(-k)\rho(-k)}\right)^{-1}e^{2i\theta(k)}&0
		\end{pmatrix},
		& k\in (-i\eta_2,-i\beta_1).
	\end{cases}
\end{align}
For later use, we remove the remaining scalar factors on the main cuts by introducing the function
$\tilde f(k)$ as follows:
\begin{align}\label{def-tildef}
	\tilde f(k)
	=\exp\Bigg\{
	\frac{R(k)}{2\pi i}
	\Bigg[
	&\int_{i\eta_1}^{i\beta_1}
	\frac{\log \dfrac{2\rho(\zeta)}{1+r(\zeta)\rho(\zeta)}}
	{R_+(\zeta)(\zeta-k)}\,d\zeta
	-\int_{i\beta_1}^{i\eta_2}
	\frac{\log \dfrac{2r(\zeta)}{1+r(\zeta)\rho(\zeta)}}
	{R_+(\zeta)(\zeta-k)}\,d\zeta
	\nonumber\\
	&-\int_{-i\beta_1}^{-i\eta_1}
	\frac{\log \dfrac{2\rho(-\zeta)}{1+r(-\zeta)\rho(-\zeta)}}
	{R_+(\zeta)(\zeta-k)}\,d\zeta
	+\int_{-i\eta_2}^{-i\beta_1}
	\frac{\log \dfrac{2r(-\zeta)}{1+r(-\zeta)\rho(-\zeta)}}
	{R_+(\zeta)(\zeta-k)}\,d\zeta
	\nonumber\\
	&+\int_{-i\eta_1}^{i\eta_1}
	\frac{i\tilde\Delta}{R(\zeta)(\zeta-k)}\,d\zeta
	\Bigg]
	\Bigg\}.
\end{align}
The constant $\tilde\Delta$ is determined by the normalization condition
$\tilde f(k)=1+\mathcal O(k^{-1})$ as $k\to\infty$. More precisely,
\begin{align}\label{def-tDelta}
	\tilde\Delta
	=
	-i\Bigg[
	&\int_{i\beta_1}^{i\eta_2}
	\frac{\log \dfrac{2r(\zeta)}{1+r(\zeta)\rho(\zeta)}}{R_+(\zeta)}\,d\zeta
	-\int_{i\eta_1}^{i\beta_1}
	\frac{\log \dfrac{2\rho(\zeta)}{1+r(\zeta)\rho(\zeta)}}{R_+(\zeta)}\,d\zeta
	\nonumber\\
	&+\int_{-i\beta_1}^{-i\eta_1}
	\frac{\log \dfrac{2\rho(-\zeta)}{1+r(-\zeta)\rho(-\zeta)}}{R_+(\zeta)}\,d\zeta
	-\int_{-i\eta_2}^{-i\beta_1}
	\frac{\log \dfrac{2r(-\zeta)}{1+r(-\zeta)\rho(-\zeta)}}{R_+(\zeta)}\,d\zeta
	\Bigg]
	\Bigg[
	\int_{-i\eta_1}^{i\eta_1}\frac{d\zeta}{R(\zeta)}
	\Bigg]^{-1}.
\end{align}
Using the symmetry of the integrands on the two cuts, the above expression reduces to
\begin{align}\label{def-tDelta-simplified}
	\tilde\Delta
	=
	\frac{\eta_2}{K(\hat m)}
	\left(
	\int_{i\eta_1}^{i\beta_1}
	\frac{\log \dfrac{2\rho(\zeta)}{1+r(\zeta)\rho(\zeta)}}{R_+(\zeta)}\,d\zeta
	-\int_{i\beta_1}^{i\eta_2}
	\frac{\log \dfrac{2r(\zeta)}{1+r(\zeta)\rho(\zeta)}}{R_+(\zeta)}\,d\zeta
	\right)\in\R.
\end{align}
By the Plemelj formula, $\tilde f$ satisfies the scalar jump relations
\begin{align}
	\tilde f_+(k)\tilde f_-(k)
	&=
	\frac{2\rho(k)}{1+r(k)\rho(k)},
	&& k\in i(\eta_1,\beta_1),\\
	\tilde f_+(k)\tilde f_-(k)
	&=
	\left(
	\frac{2r(k)}{1+r(k)\rho(k)}
	\right)^{-1},
	&& k\in i(\beta_1,\eta_2),\\
	\tilde f_+(k)\tilde f_-(k)
	&=
	\left(
	\frac{2\rho(-k)}{1+r(-k)\rho(-k)}
	\right)^{-1},
	&& k\in i(-\beta_1,-\eta_1),\\
	\tilde f_+(k)\tilde f_-(k)
	&=
	\frac{2r(-k)}{1+r(-k)\rho(-k)},
	&& k\in i(-\eta_2,-\beta_1),\\
	\frac{\tilde f_+(k)}{\tilde f_-(k)}
	&=
	e^{i\tilde\Delta},
	&& k\in i[-\eta_1,\eta_1],\\
	\tilde f(k)
	&=
	1+\mathcal O\left(\frac1k\right),
	&& k\to\infty .
\end{align}
Moreover, since $\frac{2\rho(k)}{1+r(k)\rho(k)}>0$ and $\frac{2r(k)}{1+r(k)\rho(k)}>0$ for $k\in i(\eta_1,\eta_2)$, so by the Lemma C.1 of \cite{Boutet3}, we know $\tilde f$ is  analytic and nonzero function for $k\in \C \setminus i[-\eta_2,\eta_2]$.
\begin{figure}[htbp]
	\centering
	\begin{tikzpicture}[
		scale=0.5,
		line cap=round,
		line join=round,
		>=Stealth,
		curvearrow/.style={
			thick,
			postaction={decorate},
			decoration={markings, mark=at position 0.58 with {\arrow{>}}}
		}
		]
		
		
		\def\ya{2.0}   
		\def\yp{4.2}   
		\def\yb{6.6}   
		
		\def\hs{0.82}
		\def\hb{1.15}
		
		\coordinate (Dm) at (0,-\yb);   
		\coordinate (Dp) at (0,-\yp);   
		\coordinate (Da) at (0,-\ya);   
		
		\coordinate (Ua) at (0,\ya);    
		\coordinate (Up) at (0,\yp);    
		\coordinate (Ub) at (0,\yb);    
		
		
		\draw[thick] (Dm) -- (Dp) -- (Da);
		\draw[thick] (Ua) -- (Up) -- (Ub);
		
		
		\draw[curvearrow]
		(Dm) .. controls (-\hb,-5.95) and (-\hb,-4.95) .. (Dp);
		
		\draw[curvearrow]
		(Dm) .. controls (\hb,-5.95) and (\hb,-4.95) .. (Dp);
		
		
		\draw[curvearrow]
		(Dp) .. controls (-\hs,-3.45) and (-\hs,-2.65) .. (Da);
		
		\draw[curvearrow]
		(Dp) .. controls (\hs,-3.45) and (\hs,-2.65) .. (Da);
		
		
		\draw[curvearrow]
		(Ua) .. controls (-\hs,2.65) and (-\hs,3.45) .. (Up);
		
		\draw[curvearrow]
		(Ua) .. controls (\hs,2.65) and (\hs,3.45) .. (Up);
		
		
		\draw[curvearrow]
		(Up) .. controls (-\hb,4.95) and (-\hb,5.95) .. (Ub);
		
		\draw[curvearrow]
		(Up) .. controls (\hb,4.95) and (\hb,5.95) .. (Ub);
		
		
		\node[left] at (-1.45,-5.35) {$\tilde{\mathcal{C}}_{-1}$};
		\node[left] at (-1.18,-3.05) {$\tilde{\mathcal{C}}_{-2}$};
		
		\node[left]  at (Dm) {$-i\eta_2$};
		\node[left]  at (Da) {$-i\eta_1$};
		\node[right] at (Dp) {$-i\beta_1(\xi)$};
		
		
		\node[left] at (-1.18,3.05) {$\tilde{\mathcal{C}}_{2}$};
		\node[left] at (-1.45,5.35) {$\tilde{\mathcal{C}}_{1}$};
		
		\node[left]  at (Ua) {$i\eta_1$};
		\node[left]  at (Ub) {$i\eta_2$};
		\node[right] at (Up) {$i\beta_1(\xi)$};
		
	\end{tikzpicture}
	\caption{Opening lenses for $\xi\in (\hat{\xi},\xi_{*})$.}
	\label{openinglenses2}
\end{figure}
Follow the step as section \ref{sec4.1}, we define the transformation
\begin{align}
	\tilde T(k)=\tilde Y(k) e^{i(t\hat g(k)-\theta(k))\sigma_3}\tilde f^{\sigma_3}(k).
\end{align}
The final transformation converts the jumps on the main cuts into constant
off-diagonal matrices. More precisely, the function $\tilde T(k)$ satisfies
\begin{equation}\label{jump-S-tilde}
	\tilde T_+(k)=\tilde T_-(k)\,\tilde V_T(k),
\end{equation}
where the jump matrix is given by
\begin{equation}\label{jump-matrix-T-tilde}
	\tilde V_T(k)=
	\begin{cases}
		\begin{pmatrix}
			1&0\\[1mm]
			\displaystyle
			-i\,\frac{r(k)\rho(k)-1}{2\rho(k)}
			\,\tilde f(k)^2 e^{2it\hat g(k)}&1
		\end{pmatrix},
		& k\in \tilde{\mathcal C}_2,
		\\[5mm]
		\begin{pmatrix}
			1&
			\displaystyle
			-i\,\frac{r(k)\rho(k)-1}{2r(k)}
			\,\tilde f(k)^{-2} e^{-2it\hat g(k)}\\[1mm]
			0&1
		\end{pmatrix},
		& k\in \tilde{\mathcal C}_1,
		\\[5mm]
		\begin{pmatrix}
			1&
			\displaystyle
			i\,\frac{r(-k)\rho(-k)-1}{2\rho(-k)}
			\,\tilde f(k)^{-2} e^{-2it\hat g(k)}\\[1mm]
			0&1
		\end{pmatrix},
		& k\in \tilde{\mathcal C}_{-2},
		\\[5mm]
		\begin{pmatrix}
			1&0\\[1mm]
			\displaystyle
			i\,\frac{r(-k)\rho(-k)-1}{2r(-k)}
			\,\tilde f(k)^2 e^{2it\hat g(k)}&1
		\end{pmatrix},
		& k\in \tilde{\mathcal C}_{-1},
		\\[5mm]
		\begin{pmatrix}
			0&-i\\
			-i&0
		\end{pmatrix},
		& k\in i(\eta_1,\eta_2),
		\\[5mm]
		\begin{pmatrix}
			0&i\\
			i&0
		\end{pmatrix},
		& k\in i(-\eta_2,-\eta_1),
		\\[5mm]
		\begin{pmatrix}
			e^{it\hat\Omega+i\tilde\Delta}&0\\
			0&e^{-it\hat\Omega-i\tilde\Delta}
		\end{pmatrix},
		& k\in i[-\eta_1,\eta_1].
	\end{cases}
\end{equation}
Outside fixed small disks centered at
$\pm i\eta_1$, $\pm i\eta_2$, and $\pm i\beta_1(\xi)$,
the lens contours can be chosen so that
$\tilde f^{\pm1}$ and the coefficients appearing in
\eqref{jump-matrix-T-tilde} are uniformly bounded there.
Together with the sign estimates in Lemma~\ref{lemma3.1},
this implies that the off-diagonal entries of the jumps on
these lens lips are $\mathcal O(e^{-ct})$
for some $c>0$.
 The solution \(\tilde T\) also tends to a model problem with a constant jump matrix. To simplify the notation, we still denote it by \(T^{\infty}\), and construct the corresponding matrix model \(P^{\infty}\).

Define $D_{\epsilon}(a)$ denote the small disks centered at $a$ with radius $\epsilon$. Let $\epsilon=\epsilon(\xi)$ be so small that the six disks $$D_\epsilon(\pm i\beta_1),\ \  D_\epsilon(\pm i\eta_1).\ \  D_\epsilon(\pm i\eta_2),$$
are disjoint from each other.
The  local parametrix $P^{\pm i\eta_1}$ and $P^{\pm i\eta_2}$ near the points $\pm i\eta_1$, $\pm i\eta_2$
can be constructed by using the modified Bessel functions, see \cite{GYJ}. On boundaries $\partial D_\epsilon(\pm i\eta_1)$ and $\partial D_\epsilon(\pm i\eta_2)$, we have
\begin{align*}
	P^{\pm i\eta_1}(P^{\infty})^{-1}=I+\mathcal{O}(t^{-1}),\ \ k\in \partial D_\epsilon(\pm i\eta_1),\\
	P^{\pm i\eta_2}(P^{\infty})^{-1}=I+\mathcal{O}(t^{-1}),\ \ k\in \partial D_\epsilon(\pm i\eta_2).
\end{align*}
 The construction is standard, we omit it.
 Next, we show how to construct the local parametrices in neighborhoods of $\pm i\beta_1(\xi)$.
 
\subsubsection{The local model near $i\beta_1$} 
Since $\pm i\beta_1(\xi)$ are the stationary points of $\hat g(k)$, the local parametrices in neighborhoods of $\pm i\beta_1(\xi)$ can be constructed in terms of parabolic cylinder functions.
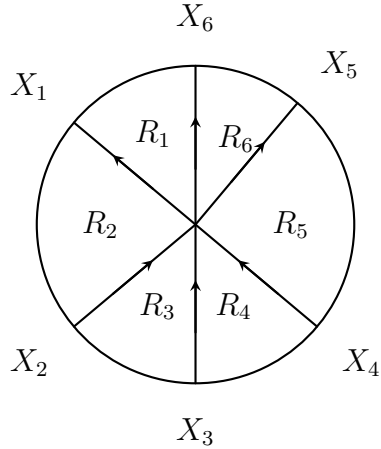
\begin{figure}[htbp]
	\centering
	\begin{tikzpicture}[scale=0.7,>=stealth,line cap=round,line join=round]
		
		\def\r{3}
		
		\coordinate (O) at (0,0);
		
		\coordinate (X1p) at (140:\r);
		\coordinate (X2p) at (220:\r);
		\coordinate (X3p) at (270:\r);
		\coordinate (X4p) at (320:\r);
		\coordinate (X5p) at (50:\r);
		\coordinate (X6p) at (90:\r);
		
		\draw[thick] (O) circle (\r);
		
		\draw[thick] (O) -- (X1p);
		\draw[thick] (O) -- (X2p);
		\draw[thick] (O) -- (X3p);
		\draw[thick] (O) -- (X4p);
		\draw[thick] (O) -- (X5p);
		\draw[thick] (O) -- (X6p);
		
		\draw[->,thick] (140:1.05) -- (140:2.04);
		\draw[->,thick] (50:1.05)  -- (50:2.04);
		\draw[->,thick] (90:1.05)  -- (90:2.04);
		
		\draw[->,thick] (220:2.04) -- (220:1.05);
		\draw[->,thick] (270:2.04) -- (270:1.05);
		\draw[->,thick] (320:2.04) -- (320:1.05);
		
		\node[above left]  at (140:3.35) {$X_1$};
		\node[below left]  at (220:3.35) {$X_2$};
		\node[below]       at (270:3.45) {$X_3$};
		\node[below right] at (320:3.35) {$X_4$};
		\node[above right] at (50:3.35)  {$X_5$};
		\node[above]       at (90:3.45)  {$X_6$};
		
		\node at (115:1.9) {$R_1$};
		\node at (180:1.8) {$R_2$};
		\node at (245:1.7) {$R_3$};
		\node at (295:1.7) {$R_4$};
		\node at (0:1.8)   {$R_5$};
		\node at (65:1.8)  {$R_6$};
		
	\end{tikzpicture}
	\caption{The contour in the disk $D_{\epsilon}(i\beta_1)$ and the sets $\{R_j\}_{1}^{6}$}
	\label{fig:disk-regions}
\end{figure}
 We define the $D_{\epsilon\pm }(i\beta_1)$ the semicircles on the left and right parts of the complex plane
with respect to the imaginaryaxis. Let
\[
X_{1}\cup X_5=\tilde{\mathcal{C}}_1\cap D_{\epsilon}(i\beta_1),\qquad
X_{2}\cup X_4=\tilde{\mathcal{C}}_2\cap D_{\epsilon}(i\beta_1),\qquad
X_3\cup X_6=i(\eta_1,\eta_2)\cap D_{\epsilon}(i\beta_1).
\]
These curves divide the disk $D_{\epsilon}(i\beta_1)$ into six regions $\{R_j\}_{j=1}^6$; see Fig.~\ref{fig:disk-regions}.
The expansion of $\hat g(k)$ near the point $k=i\beta_1$ has the form
\begin{align}\label{gexp}
	\hat g(k)=\pm (\hat g(i\beta_1)+\frac12 \hat g_+''(i\beta_1)(k-i\beta_1)^2+\mathcal{O}(k-i\beta_1)^3),\ \ k\to i\beta_1,\ k\in D_{\epsilon\pm }(i\beta_1),
\end{align}
where $$\hat g_+''(i\beta_1;\xi)=\begin{cases}
	\displaystyle
	\omega\xi
	\frac{-2\beta_1(\beta_1^2-\beta_0^2)(\beta_1^2+\gamma^2)}
	{\left(-\beta_1^2+\frac14\right)^2
		\left|\sqrt{(-\beta_1^2+\eta_1^2)(-\beta_1^2+\eta_2^2)}_+\right|}
	>0,
	& \xi\in(\hat{\xi},0),
	\\[3mm]
	\displaystyle
	\omega\xi
	\frac{2\beta_1(\beta_1^2-\beta_0^2)(\beta_2^2-\beta_1^2)}
	{\left(-\beta_1^2+\frac14\right)^2
		\left|\sqrt{(-\beta_1^2+\eta_1^2)(-\beta_1^2+\eta_2^2)}_+\right|}
	>0,
	& \xi\in(0,\xi_*).
\end{cases}.$$
  Define $g_{i\beta_1}(k)$ for $k$ near $i\beta_1$ by
\begin{align}
	g_{i\beta_1}(k)=
	\begin{cases}
		\hat g(k)-\hat g_{+}(i\beta_1),
		& k\in R_1\cup R_2\cup R_3,\\[1mm]
		-\bigl(\hat g(k)-\hat g_{-}(i\beta_1)\bigr),
		& k\in R_4\cup R_5\cup R_6.
	\end{cases}
\end{align}
We define the function $M^{(1)}(k)$ for $k$ near $i\beta_1$ by
\begin{align}
	M^{(1)}(k)
	=
	\tilde{T}(k)\tilde{f}^{-\sigma_3}(k)
	,
	\qquad
	k\in D_{\epsilon}(i\beta_1)\setminus \bigcup_{j=1}^{6}X_j .
\end{align}
Across the curves $\bigcup_{j=1}^{6}X_j$, $M^{(1)}(k)$ satisfies the jump condition
\[
M^{(1)}_{+}(k)=M^{(1)}_{-}(k)V^{(1)}(k),
\]
where
\begin{align}
	V^{(1)}(k)=
	\begin{cases}
		\begin{pmatrix}
			1 & -\displaystyle i\frac{r(k)\rho(k)-1}{2r(k)}e^{-2it\hat g(k)}\\[2mm]
			0 & 1
		\end{pmatrix},
		& k\in X_1\cup X_5,\\[5mm]
		\begin{pmatrix}
			1 & 0\\[2mm]
			-\displaystyle i\frac{r(k)\rho(k)-1}{2\rho(k)}e^{2it\hat g(k)} & 1
		\end{pmatrix},
		& k\in X_2\cup X_4,\\[5mm]
		\begin{pmatrix}
			0 & -\displaystyle i\frac{2\rho(k)}{1+r(k)\rho(k)}\\[2mm]
			-\displaystyle i\left(\frac{2\rho(k)}{1+r(k)\rho(k)}\right)^{-1} & 0
		\end{pmatrix},
		& k\in X_3,\\[5mm]
		\begin{pmatrix}
			0 & -\displaystyle i\left(\frac{2r(k)}{1+r(k)\rho(k)}\right)^{-1}\\[2mm]
			-\displaystyle i\frac{2r(k)}{1+r(k)\rho(k)} & 0
		\end{pmatrix},
		& k\in X_6.
	\end{cases}
\end{align}
Here all contours are oriented to the right as in Fig.~\ref{fig:disk-regions}.
We next introduce a local scalar function which removes the nonconstant
factors appearing on the two subintervals adjacent to the stationary point
$i\beta_1$. Define
\begin{align}\label{def-local-delta}
	\delta(k)
	=
	\exp\left\{
	\frac{1}{2\pi i}
	\left(
	\int_{i\eta_1}^{i\beta_1}
	\log\left(\frac{2\rho(\zeta)}{1+r(\zeta)\rho(\zeta)}\right)
	\frac{d\zeta}{\zeta-k}
	-
	\int_{i\beta_1}^{i\eta_2}
	\log\left(\frac{2r(\zeta)}{1+r(\zeta)\rho(\zeta)}\right)
	\frac{d\zeta}{\zeta-k}
	\right)
	\right\}.
\end{align}
The branches of the logarithms are chosen continuously on the corresponding
subintervals of the cut. Since
\[
\log\left(\frac{2\rho(i\beta_1)}{1+r(i\beta_1)\rho(i\beta_1)}\right)\in \R,
\qquad
\log\left(\frac{2r(i\beta_1)}{1+r(i\beta_1)\rho(i\beta_1)}\right)\in \R,
\]
Lemma~C.1 of \cite{Boutet3} implies that $\delta(k)$ is bounded and analytic in
$
D_\epsilon(i\beta_1)\setminus i(\eta_1,\eta_2).
$
Moreover, by the Plemelj formula, the boundary values of $\delta$ satisfy
\begin{align}
	\frac{\delta_+(k)}{\delta_-(k)}
	=
	\begin{cases}
		\displaystyle
		\frac{2\rho(k)}{1+r(k)\rho(k)},
		& k\in i(\eta_1,\beta_1),
		\\[3mm]
		\displaystyle
		\left(
		\frac{2r(k)}{1+r(k)\rho(k)}
		\right)^{-1},
		& k\in i(\beta_1,\eta_2).
	\end{cases}
\end{align}
We next introduce two piecewise defined functions $A(k)$ and $B(k)$ by
\begin{align}
	A(k)
	&=
	\begin{cases}
		\delta^{\sigma_3}(k), & k\in R_1\cup R_2\cup R_3,\\
		\delta^{-\sigma_3}(k), & k\in R_4\cup R_5\cup R_6,
	\end{cases}
	\\
	B(k)
	&=
	\begin{cases}
		\begin{pmatrix}
			0&i\\
			i&0
		\end{pmatrix}, & k\in R_1\cup R_2\cup R_3,\\[3mm]
		I, & k\in R_4\cup R_5\cup R_6.
	\end{cases}
\end{align}
Define
\begin{align}
	M^{(2)}(k)=M^{(1)}(k)A(k)B(k)e^{it\hat g_{+}(i\beta_1)\sigma_3}.
\end{align}
With this transformation, the jumps on $X_3$ and $X_6$ are removed. The function
$M^{(2)}$ therefore satisfies
\begin{align}
	M^{(2)}_+(k)=M^{(2)}_-(k)
	\begin{cases}
		\begin{pmatrix}
			1&0\\[1mm]
			\displaystyle
			-i\frac{r(k)\rho(k)-1}{2r(k)}
			\delta^{-2}(k)e^{-2i tg_{i\beta_1}(k)}&1
		\end{pmatrix},
		& k\in X_1,
		\\[5mm]
		\begin{pmatrix}
			1&
			\displaystyle
			-i\frac{r(k)\rho(k)-1}{2r(k)}
			\delta^{2}(k)e^{2it g_{i\beta_1}(k)}
			\\[1mm]
			0&1
		\end{pmatrix},
		& k\in X_5,
		\\[5mm]
		\begin{pmatrix}
			1&
			\displaystyle
			-i\frac{r(k)\rho(k)-1}{2\rho(k)}
			\delta^{2}(k)e^{2i tg_{i\beta_1}(k)}
			\\[1mm]
			0&1
		\end{pmatrix},
		& k\in X_2,
		\\[5mm]
		\begin{pmatrix}
			1&0\\[1mm]
			\displaystyle
			-i\frac{r(k)\rho(k)-1}{2\rho(k)}
			\delta^{-2}(k)e^{-2i tg_{i\beta_1}(k)}&1
		\end{pmatrix},
		& k\in X_4,
		\\[5mm]
		I,
		& k\in X_3\cup X_6.
	\end{cases}
\end{align}

To identify this local RH problem with the model problem $m^X$ in
Appendix~A of \cite{Boutet3}, we introduce a conformal coordinate near
$i\beta_1$. Namely, for $k\in D_\epsilon(i\beta_1)$, set
\[
\zeta=\zeta(k):=\sqrt{-tg_{i\beta_1,+}(k)}.
\]
The branch of the square root is chosen as follows. Since
$g_{i\beta_1,+}(k)$ has a double zero at $k=i\beta_1$, we may write
\[
\zeta=-i\sqrt{t}(k-i\beta_1)\psi_{i\beta_1}(k),
\]
where $\psi_{i\beta_1}$ is analytic in $D_\epsilon(i\beta_1)$. By the
expansion of $g$ in \eqref{gexp}, one has
$
g_{i\beta_1,+}''(i\beta_1)\in \mathbb R_+.
$
We fix the branch by requiring
$
\psi_{i\beta_1}(i\beta_1)>0.
$

For $\epsilon>0$ sufficiently small, the map $k\mapsto\zeta(k)$ is a
biholomorphism from $D_\epsilon(i\beta_1)$ onto a neighborhood of the origin.
Moreover, it maps the contours $X_3$ and $X_6$ to the rays $R_-$ and $R_+$,
respectively. 
We may assume that $X_1$, $X_2$, $X_4$, and $X_5$ are mapped to the straight
rays in the $\zeta$-plane with arguments
\[
\arg\zeta=\frac{\pi}{4},\qquad
\arg\zeta=\frac{3\pi}{4},\qquad
\arg\zeta=-\frac{3\pi}{4},\qquad
\arg\zeta=-\frac{\pi}{4},
\]
respectively.

We now describe the behavior of the local scalar function $\delta(k)$ as
$k\to i\beta_1$. Let
\[
\log_{i\eta_1}(k-i\beta_1),\qquad
\log_{i\eta_2}(k-i\beta_1)
\]
denote the logarithm of $k-i\beta_1$ with branch cuts along
$i(\eta_1,\beta_1)$ and $i(\beta_1,\eta_2)$, respectively. We also introduce
the auxiliary logarithmic functions
\begin{align*}
	L_{\eta_1}(s,k)&=\log(k-s),
	&& s\in i(\eta_1,\beta_1),\quad
	k\in D_\epsilon(i\beta_1)\setminus i(\eta_1,\beta_1),\\
	L_{\eta_2}(s,k)&=\log(k-s),
	&& s\in i(\beta_1,\eta_2),\quad
	k\in D_\epsilon(i\beta_1)\setminus i(\beta_1,\eta_2).
\end{align*}
The branches are chosen so that, for each fixed $k$, the functions
$L_{\eta_1}(s,k)$ and $L_{\eta_2}(s,k)$ are continuous in $s$ on their
respective intervals, and such that
\[
L_{\eta_1}(i\beta_1,k)=\log_{i\eta_1}(k-i\beta_1),
\qquad
L_{\eta_2}(i\beta_1,k)=\log_{i\eta_2}(k-i\beta_1).
\]

Set
\[
\ell_\rho(s):=
\log\left(\frac{2\rho(s)}{1+r(s)\rho(s)}\right),
\qquad
\ell_r(s):=
\log\left(\frac{2r(s)}{1+r(s)\rho(s)}\right).
\]
Subtracting the values of these functions at $i\beta_1$ in the
Cauchy integrals, we obtain
\begin{align}\label{delta-local-representation}
	\delta(k)
	=
	\exp\left\{
	\nu_\rho\log_{i\eta_1}(k-i\beta_1)
	+\nu_r\log_{i\eta_2}(k-i\beta_1)
	+\chi(k)
	\right\},
\end{align}
where
\begin{align}\label{def-nu-rho-r}
	\nu_\rho=\frac{\ell_\rho(i\beta_1)}{2\pi i},\ 
	\nu_r=\frac{\ell_r(i\beta_1)}{2\pi i},
\end{align}
and
\begin{align}\label{def-chi-local}
	\chi(k)=\frac{1}{2\pi i}\Bigg[
	&-\ell_\rho(i\beta_1)L_{\eta_1}(i\eta_1,k)
	-\ell_r(i\beta_1)L_{\eta_2}(i\eta_2,k)
	\nonumber\\
	&+\int_{i\eta_1}^{i\beta_1}
	\frac{\ell_\rho(s)-\ell_\rho(i\beta_1)}{s-k}\,ds
	-\int_{i\beta_1}^{i\eta_2}
	\frac{\ell_r(s)-\ell_r(i\beta_1)}{s-k}\,ds
	\Bigg].
\end{align}
The integrals in \eqref{def-chi-local} converge at the spectral
endpoints. Moreover, $\chi$ is analytic on each side of the cut and
has a common finite limit at $i\beta_1$, denoted by $\chi(i\beta_1)$.
As $k\to i\beta_1$ within either side of the cut,
\[
\chi(k)-\chi(i\beta_1)
=
\mathcal O\left(
|k-i\beta_1|
\log\frac{1}{|k-i\beta_1|}
\right).
\]

Let $\log_a \zeta$ denote the logarithm with branch cut along
$\arg \zeta=a$, that is,
\[
\log_a\zeta=\log|\zeta|+i\arg\zeta,
\qquad \arg\zeta\in(a,a+2\pi].
\]
Thus $\log_{-\pi}\zeta$ is the principal logarithm. Since the conformal map
$k\mapsto \zeta$ sends the two sides of the cut through $i\beta_1$ to the
rays $R_-$ and $R_+$, respectively, we have
\begin{align}
	\log_{i\eta_1}(k-i\beta_1)
	&=
	\log_{i\eta_1}
	\left(
	\frac{i\zeta}{\sqrt{t}\,\psi_{i\beta_1}(k)}
	\right)
	=
	\frac{i\pi}{2}-\frac12\log t+\log\zeta-\log\psi_{i\beta_1}(k),
	\nonumber\\
	&\hspace{7cm}
	k\in D_\epsilon(i\beta_1)\setminus i(\eta_1,\beta_1),
	\\
	\log_{i\eta_2}(k-i\beta_1)
	&=
	\log_{i\eta_2}
	\left(
	\frac{i\zeta}{\sqrt{t}\,\psi_{i\beta_1}(k)}
	\right)
	=
	\frac{i\pi}{2}-\frac12\log t+\log_0\zeta-\log\psi_{i\beta_1}(k),
	\nonumber\\
	&\hspace{7cm}
	k\in D_\epsilon(i\beta_1)\setminus i(\beta_1,\eta_2).
\end{align}

We now separate the singular $\zeta$-dependence from the analytic prefactor.
Define
\begin{align}\label{def-p-zeta}
	p(\zeta)
	=
	\exp\left\{
	\nu_\rho\log\zeta+\nu_r\log_0\zeta
	\right\},
	\qquad \zeta\in\mathbb C\setminus\mathbb R .
\end{align}
Furthermore, set
\begin{align}\label{def-delta0-delta1}
	\delta_0
	&=
	\exp\left\{
	(\nu_\rho+\nu_r)
	\left(
	\frac{i\pi}{2}-\frac12\log t-\log\psi_{i\beta_1}(i\beta_1)
	\right)
	+\chi(i\beta_1)
	\right\},
	\\
	\delta_1(k)
	&=
	\exp\left\{
	-(\nu_\rho+\nu_r)
	\left(
	\log\psi_{i\beta_1}(k)-\log\psi_{i\beta_1}(i\beta_1)
	\right)
	+\chi(k)-\chi(i\beta_1)
	\right\}.
\end{align}
Thus, in $D_\epsilon(i\beta_1)$, the function $\delta(k)$ can be written in
the factorized form
\begin{align}\label{delta-factorization-local}
	\delta(k)=\delta_0(t)\,\delta_1(k)\,p(\zeta(k)),
\end{align}
where $\delta_1$ is analytic on each side of the cut, has the common
limit $\delta_1(i\beta_1)=1$, and satisfies
\[
\delta_1(k)-1
=
\mathcal O\left(
|k-i\beta_1|
\log\frac{1}{|k-i\beta_1|}
\right),
\qquad k\to i\beta_1.
\]

Define $M^{(3)}$ by
\begin{align}\label{def-M3}
	M^{(3)}(\zeta(k))
	=
	M^{(2)}(k)\,
	\delta_0^{\sigma_3}
	\left(\frac{r(i\beta_1)}{\rho(i\beta_1)}\right)^{\frac14\sigma_3}.
\end{align}
Here and below, whenever a function of $k$ is written as a function of
$\zeta$, we use the inverse relation $k=k(\zeta)$. Then $M^{(3)}$ satisfies
the jump condition
\begin{align}
	M^{(3)}_+(\zeta)=M^{(3)}_-(\zeta)V^{(3)}(\zeta),
\end{align}
where
\begin{align}\label{V3-before-ptilde}
	V^{(3)}(\zeta)=
	\begin{cases}
		\begin{pmatrix}
			1&0\\[1mm]
			\displaystyle
			i\sqrt{\frac{r(i\beta_1)}{\rho(i\beta_1)}}
			\frac{r(k)\rho(k)-1}{2r(k)}
			\delta_1(k)^{-2}p(\zeta)^{-2}e^{2i\zeta^2}
			&1
		\end{pmatrix},
		& \arg\zeta=\dfrac{\pi}{4},
		\\[6mm]
		\begin{pmatrix}
			1&
			\displaystyle
			i\sqrt{\frac{\rho(i\beta_1)}{r(i\beta_1)}}
			\frac{r(k)\rho(k)-1}{2r(k)}
			\delta_1(k)^2p(\zeta)^2e^{-2i\zeta^2}
			\\[1mm]
			0&1
		\end{pmatrix},
		& \arg\zeta=-\dfrac{\pi}{4},
		\\[6mm]
		\begin{pmatrix}
			1&
			\displaystyle
			-i\sqrt{\frac{\rho(i\beta_1)}{r(i\beta_1)}}
			\frac{r(k)\rho(k)-1}{2\rho(k)}
			\delta_1(k)^2p(\zeta)^2e^{-2i\zeta^2}
			\\[1mm]
			0&1
		\end{pmatrix},
		& \arg\zeta=\dfrac{3\pi}{4},
		\\[6mm]
		\begin{pmatrix}
			1&0\\[1mm]
			\displaystyle
			-i\sqrt{\frac{r(i\beta_1)}{\rho(i\beta_1)}}
			\frac{r(k)\rho(k)-1}{2\rho(k)}
			\delta_1(k)^{-2}p(\zeta)^{-2}e^{2i\zeta^2}
			&1
		\end{pmatrix},
		& \arg\zeta=-\dfrac{3\pi}{4}.
	\end{cases}
\end{align}
All rays are oriented toward the origin.

To put this jump matrix into the standard form of the model problem in
Appendix~A of \cite{Boutet3}, introduce
\begin{align}\label{def-ptilde}
	\widetilde p(\zeta)
	=
	\exp\left\{i\nu\log_{-\frac{\pi}{2}}\zeta\right\},
\end{align}
where
\begin{align}\label{nu}
	\nu=\nu(i\beta_1)
	=
	-\frac{1}{2\pi}
	\log\left(
	\frac{2\rho(i\beta_1)}
	{1+r(i\beta_1)\rho(i\beta_1)}
	\right)
	-\frac{1}{2\pi}
	\log\left(
	\frac{2r(i\beta_1)}
	{1+r(i\beta_1)\rho(i\beta_1)}
	\right)>0.
\end{align}
The positivity of $\nu$ follows from
\[
0<\frac{4r(i\beta_1)\rho(i\beta_1)}
{\left(1+r(i\beta_1)\rho(i\beta_1)\right)^2}<1.
\]
By the definition of $p$, one has
\begin{align}\label{p-ptilde-relation}
	p(\zeta)
	=
	\widetilde p(\zeta)
	\begin{cases}
		1,
		& \arg\zeta\in(0,\pi),
		\\[1mm]
		\displaystyle
		\frac{1+r(i\beta_1)\rho(i\beta_1)}{2\rho(i\beta_1)},
		& \arg\zeta\in\left(-\pi,-\dfrac{\pi}{2}\right),
		\\[3mm]
		\displaystyle
		\frac{2r(i\beta_1)}
		{1+r(i\beta_1)\rho(i\beta_1)},
		& \arg\zeta\in\left(-\dfrac{\pi}{2},0\right).
	\end{cases}
\end{align}
Therefore, the jump matrix may equivalently be written as
\begin{align}\label{V3-after-ptilde}
	V^{(3)}(\zeta)=
	\begin{cases}
		\begin{pmatrix}
			1&0\\[1mm]
			\displaystyle
			i\sqrt{\frac{r(i\beta_1)}{\rho(i\beta_1)}}
			\frac{r(k)\rho(k)-1}{2r(k)}
			\delta_1(k)^{-2}\widetilde p(\zeta)^{-2}e^{2i\zeta^2}
			&1
		\end{pmatrix},
		& \arg\zeta=\dfrac{\pi}{4},
		\\[6mm]
		\begin{pmatrix}
			1&
			\displaystyle
			i\sqrt{\frac{\rho(i\beta_1)}{r(i\beta_1)}}
			\frac{r(k)\rho(k)-1}{2r(k)}
			\left(
			\frac{2r(i\beta_1)}
			{1+r(i\beta_1)\rho(i\beta_1)}
			\right)^2
			\delta_1(k)^2\widetilde p(\zeta)^2e^{-2i\zeta^2}
			\\[1mm]
			0&1
		\end{pmatrix},
		& \arg\zeta=-\dfrac{\pi}{4},
		\\[6mm]
		\begin{pmatrix}
			1&
			\displaystyle
			-i\sqrt{\frac{\rho(i\beta_1)}{r(i\beta_1)}}
			\frac{r(k)\rho(k)-1}{2\rho(k)}
			\delta_1(k)^2\widetilde p(\zeta)^2e^{-2i\zeta^2}
			\\[1mm]
			0&1
		\end{pmatrix},
		& \arg\zeta=\dfrac{3\pi}{4},
		\\[6mm]
		\begin{pmatrix}
			1&0\\[1mm]
			\displaystyle
			-i\sqrt{\frac{r(i\beta_1)}{\rho(i\beta_1)}}
			\frac{r(k)\rho(k)-1}{2\rho(k)}
			\left(
			\frac{2\rho(i\beta_1)}
			{1+r(i\beta_1)\rho(i\beta_1)}
			\right)^2
			\delta_1(k)^{-2}\widetilde p(\zeta)^{-2}e^{2i\zeta^2}
			&1
		\end{pmatrix},
		& \arg\zeta=-\dfrac{3\pi}{4}.
	\end{cases}
\end{align}

Set
\begin{align}\label{def-q-ibeta}
	q(i\beta_1)
	:=
	i\frac{r(i\beta_1)\rho(i\beta_1)-1}
	{2\sqrt{r(i\beta_1)\rho(i\beta_1)}}.
\end{align}
Then, for fixed $\zeta$, as the large parameter tends to infinity,
\[
i\sqrt{\frac{r(i\beta_1)}{\rho(i\beta_1)}}
\frac{r(k(\zeta))\rho(k(\zeta))-1}{2r(k(\zeta))}
\longrightarrow q(i\beta_1),
\qquad
\delta_1(k(\zeta))^{\pm2}\longrightarrow 1.
\]
Thus $V^{(3)}(\zeta)$ tends to the jump matrix $v^X$ of the model RH problem
in Appendix~A of \cite{Boutet3}.

Consequently, the local parametrix in $D_\epsilon(i\beta_1)$ is defined by
\begin{align}\label{local-parametrix-ibeta}
	P^{(i\beta_1)}(k)
	=
	Y_{i\beta_1}(k)\,
	m^X\bigl(q(i\beta_1),\zeta(k)\bigr)\,
	\delta_0^{-\sigma_3}
	\left(\frac{r(i\beta_1)}{\rho(i\beta_1)}\right)^{-\frac14\sigma_3}
	e^{-it\hat g_+(i\beta_1)\sigma_3}B(k)^{-1}A(k)^{-1}
	\widetilde f(k)^{\sigma_3},
\end{align}
where $m^X$ denotes the solution of the model RH problem in Appendix~A of
\cite{Boutet3}. The prefactor $Y_{i\beta_1}$ is analytic in
$D_\epsilon(i\beta_1)$ and is determined by the matching condition
\[
P^{(i\beta_1)}(k)\bigl(P^{(\infty)}(k)\bigr)^{-1}
=
I+\mathcal O(t^{-1/2}),
\qquad k\in\partial D_\epsilon(i\beta_1),
\qquad t\to\infty.
\]
Accordingly, we choose
\begin{align}\label{def-Y-ibeta}
	Y_{i\beta_1}(k)
	=
	P^{(\infty)}(k)\,
	\widetilde f(k)^{-\sigma_3}
	A(k)B(k)\,e^{it\hat g_+(i\beta_1)\sigma_3}
	\delta_0^{\sigma_3}
	\left(\frac{r(i\beta_1)}{\rho(i\beta_1)}\right)^{\frac14\sigma_3}.
\end{align}
Indeed, $P^{(\infty)}(k)\widetilde f(k)^{-\sigma_3}$ and $B^{-1}(k)A^{-1}(k)$ have
the same jumps across $X_3\cup X_6$. Hence the possible discontinuities cancel,
and $Y_{i\beta_1}$ is analytic in $D_\epsilon(i\beta_1)$. The local parametrix $P^{-i\beta_1}$ can be construct by the symmetry
\begin{align}
	P^{-i\beta_1}(k)=\begin{pmatrix}
		0&1\\1&0
	\end{pmatrix}
		P^{i\beta_1}(-k)
	\begin{pmatrix}
		0&1\\1&0
	\end{pmatrix}.
\end{align}
Let $\mathcal X$ be the union of the cross
$$ \mathcal X=X_1\cup X_2 \cup X_4 \cup X_5. $$
\begin{lemma}
	The function $P^{i\beta_1}(k)$ defined in \eqref{local-parametrix-ibeta} is an analytic and bounded function of $k\in D_\epsilon(i\beta_1)\setminus (\mathcal{X}\cup X_3\cup X_6)$. Moreover, the prefactor $Y_{i\beta_1}(k;x,t)$ is analytic and bounded in $D_\epsilon(i\beta_1)$. Across $\mathcal{X}\cup X_3\cup X_6$, the function $P^{i\beta_1}(k)$ satisfies the jump condition $P_{+}^{i\beta_1}(k)=P_-^{i\beta_1}(k)V^{i\beta_1}$, where the jump matrix $V^{i\beta_1}$ satisfies
	$$V^{i\beta_1}=\tilde V_{T},\ \ \ k\in X_3\cup X_6 ,$$
	and 
	\begin{equation}\label{eq:vmu-estimates}
		\begin{aligned}
			\|\tilde V_{T}-V^{i\beta_1}\|_{L^1(\mathcal X)}
			&=\mathcal O(t^{-1}\ln t),\\
			\|\tilde V_{T}-V^{i\beta_1}\|_{L^2( \mathcal X)}
			&=\mathcal O(t^{-3/4}\ln t),\\
			\|\tilde V_{T}-V^{i\beta_1}\|_{L^\infty(\mathcal X)}
			&=\mathcal O(t^{-1/2}\ln t).
		\end{aligned}
	\end{equation}
	Furthermore, on \(\partial D_\epsilon(i\beta_1)\),
	\begin{equation}\label{eq:mmod-mmu-bound}
		\bigl\|P^{\infty}(P^{i\beta_1})^{-1}-I\bigr\|_{L^\infty(\partial D_\epsilon(i\beta_1))}
		=
		\mathcal O(t^{-1/2}),
		\qquad t\to+\infty,
	\end{equation}
	and
	\begin{equation}\label{eq:mmod-mmu-integral}
		P^{\infty}(P^{i\beta_1})^{-1}-I
		=
		\frac{Y_{i\beta_1}(x,t,k)m_1^X Y_{i\beta_1}(x,t,k)^{-1}}
		{\sqrt t\,(k-i\beta_1)\psi_{i\beta_1}(k)}
		+\mathcal O(t^{-1}),\ t\to\infty,\ \ k\in \partial D_\epsilon(i\beta_1).
	\end{equation}
	where
	\begin{equation}\label{eq:mX1-def}
		m_1^X=
		\begin{pmatrix}
			0&-e^{-\pi\nu}\beta^X(q)\\
			e^{\pi\nu}\overline{\beta^X(q)}&0
		\end{pmatrix},
	\end{equation}
	with $\beta^X(q)$ given by 
	\begin{equation}\label{A.5}
		\beta^{X}(q)
		:=
		\frac{\sqrt{\nu(q)}}{2}
		\exp\left\{
		i\left(
		-\frac{3\pi}{4}
		-2\nu(q)\ln 2
		-\arg q
		+\arg\Gamma(i\nu(q))
		\right)
		\right\}.
	\end{equation}
\end{lemma}
\begin{proof}
	The proof is similar as Lemma 6.5 of \cite{Boutet3}.
\end{proof}

\subsubsection{ Derivation of $u(x,t)$ for large  $t$ and small norm argument }
Let $\mathcal D$ denote the union of the six open disks, and $\partial D$ its boundary:
$$\mathcal D=D_\epsilon(\pm i\beta_1) \cup D_\epsilon(\pm i\eta_1) \cup D_\epsilon(\pm i\eta_2).$$
Define the error vector
\begin{gather}
	\mathcal{E}(k) = \tilde T(k) \left( P(k) \right)^{-1},
	\label{error_def}
	\shortintertext{where the global parametrix $P(k)$ is defined by}
	P(k) = \begin{cases}
		P^{\infty}(k),  & k\in\mathbb{C}\setminus (D_\epsilon{(\pm i\eta_1)}\cup D_\epsilon{(\pm i\eta_2)}\cup D_\epsilon{(\pm i\beta_1)}),  \\
		P^{\pm i\eta_{2}}(k), & k\in D_\epsilon{(\pm i\eta_2)},\\
		P^{\pm i\eta_{1}} (k), & k\in D_\epsilon{(\pm i\eta_1)}, \\
		P^{\pm i\beta_{1}}(k), & k\in D_\epsilon{(\pm i\beta_1)}.
	\end{cases}
	\label{global_p1}
\end{gather}
By the same argument as in
\cite{Girotti-1}, the possible singularity of \(\mathcal E(k)\) at \(k=0\) is removable.

We next show that \(\mathcal E(k)-(1,1)\) is small as \(t\to+\infty\). Let $\Sigma_{\mathcal{E}}=((\tilde{\mathcal{C}}_{\pm 1}\cup\tilde{\mathcal{C}}_{\pm 2})\setminus \bar{\mathcal D})\cup \partial \mathcal{D} \cup \mathcal{X}\cup \hat{\mathcal{X}}$, which is displayed in Fig.~\ref{fig:local-contour}. Then \(\mathcal E\) satisfies the RH problem
\begin{equation}\label{E_RHP}
	\begin{cases}
		\mathcal E(k)\text{ is analytic for }k\in\mathbb C\setminus\Sigma_{\mathcal E},\\[1mm]
		\mathcal E_+(k)=\mathcal E_-(k)V_{\mathcal E}(k),
		\qquad k\in\Sigma_{\mathcal E},\\[1mm]
		\mathcal E(k)=(1,1)+\mathcal O(k^{-1}),
		\qquad k\to\infty.
	\end{cases}
\end{equation}
The jump matrix \(V_{\mathcal E}\) is given by
\begin{equation}\label{VE_def}
	V_{\mathcal E}(k)=
	\begin{cases}
		P^\infty_-(k)V_{\widetilde T}(k)P^\infty_+(k)^{-1},
		& k\in
		\bigl(\widetilde{\mathcal C}_{\pm1}\cup
		\widetilde{\mathcal C}_{\pm2}\bigr)
		\setminus \mathcal D,
		\\[1mm]
		P^\infty(k)\bigl(P^{\pm i\eta_2}(k)\bigr)^{-1},
		& k\in\partial D_\epsilon(\pm i\eta_2),
		\\[1mm]
		P^\infty(k)\bigl(P^{\pm i\eta_1}(k)\bigr)^{-1},
		& k\in\partial D_\epsilon(\pm i\eta_1),
		\\[1mm]
		P^\infty(k)\bigl(P^{\pm i\beta_1}(k)\bigr)^{-1},
		& k\in\partial D_\epsilon(\pm i\beta_1),
		\\[1mm]
		P^{ i\beta_1}_-(k)V_{\tilde T}(k)
		\bigl(P^{ i\beta_1}_+(k)\bigr)^{-1},
		& k\in\mathcal X,
		\\[1mm]
		P^{ -i\beta_1}_-(k)V_{\tilde T}(k)
		\bigl(P^{ -i\beta_1}_+(k)\bigr)^{-1},
		& k\in \hat{\mathcal X}.
	\end{cases}
\end{equation}
Here \(V_{\tilde T}\) denotes the jump matrix of \(\tilde T\) and $\hat{\mathcal X}$ is the cross in $D_\epsilon{(-i\beta_1)}$.
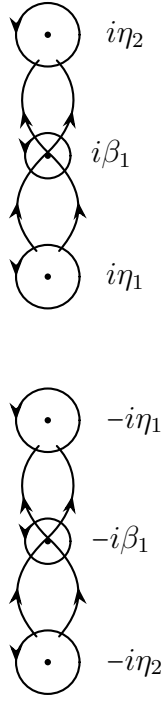
\begin{figure}[htbp]
	\centering
	\begin{tikzpicture}[
		scale=1,
		line cap=round,
		line join=round,
		thick,
		>={Stealth}
		]
		
		\tikzset{
			onearrow/.style={
				postaction={decorate},
				decoration={
					markings,
					mark=at position 0.5 with {\arrow{Stealth}}
				}
			},
			twoarrows/.style={
				postaction={decorate},
				decoration={
					markings,
					mark=at position 0.25 with {\arrow{Stealth}},
					mark=at position 0.75 with {\arrow{Stealth}}
				}
			}
		}
		
		
		\draw[onearrow] (0.42,3.2) arc[start angle=0,end angle=360,radius=0.42];
		\draw[onearrow] (0.30,1.6) arc[start angle=0,end angle=360,radius=0.30];
		\draw[onearrow] (0.42,0.0) arc[start angle=0,end angle=360,radius=0.42];
		
		\fill (0,3.2) circle (1.2pt);
		\fill (0,1.6) circle (1.2pt);
		\fill (0,0.0) circle (1.2pt);
		
		\draw[twoarrows]
		(-0.15,0.34)
		.. controls (-0.55,0.80) and (-0.45,1.20) .. (-0.05,1.60)
		.. controls (0.35,1.95) and (0.50,2.45) .. (0.12,2.86);
		
		\draw[twoarrows]
		(0.15,0.34)
		.. controls (0.55,0.80) and (0.45,1.20) .. (0.05,1.60)
		.. controls (-0.35,1.95) and (-0.50,2.45) .. (-0.12,2.86);
		
		\node[right] at (0.62,3.20) {$i\eta_2$};
		\node[right] at (0.42,1.62) {$i\beta_1$};
		\node[right] at (0.62,0.00) {$i\eta_1$};
		
		\begin{scope}[yshift=-5.1cm]
			
			\draw[onearrow] (0.42,3.2) arc[start angle=0,end angle=360,radius=0.42];
			\draw[onearrow] (0.30,1.6) arc[start angle=0,end angle=360,radius=0.30];
			\draw[onearrow] (0.42,0.0) arc[start angle=0,end angle=360,radius=0.42];
			
			\fill (0,3.2) circle (1.2pt);
			\fill (0,1.6) circle (1.2pt);
			\fill (0,0.0) circle (1.2pt);
			
			\draw[twoarrows]
			(-0.15,0.34)
			.. controls (-0.55,0.80) and (-0.45,1.20) .. (-0.05,1.60)
			.. controls (0.35,1.95) and (0.50,2.45) .. (0.12,2.86);
			
			\draw[twoarrows]
			(0.15,0.34)
			.. controls (0.55,0.80) and (0.45,1.20) .. (0.05,1.60)
			.. controls (-0.35,1.95) and (-0.50,2.45) .. (-0.12,2.86);
			
			\node[right] at (0.62,3.20) {$-i\eta_1$};
			\node[right] at (0.42,1.62) {$-i\beta_1$};
			\node[right] at (0.62,0.00) {$-i\eta_2$};
			
		\end{scope}
		
	\end{tikzpicture}
	\caption{The local contour configuration.}
	\label{fig:local-contour}
\end{figure}
Note that the solution $\mathcal E$ which we have constructed obeys the symmetry 
\begin{align}
\mathcal E(-k)=\mathcal E (k)\begin{pmatrix}
	0&1\\1&0
\end{pmatrix}.
\end{align}
Indeed, the jump matrices $V_{\mathcal E}(k)$ all satisfy the symmetry
\begin{align}\label{sym}
	V_{\mathcal E}(-k)= \begin{pmatrix}
		0&1\\1&0
	\end{pmatrix}V_{\mathcal E}(k)\begin{pmatrix}
	0&1\\1&0
	\end{pmatrix}.
\end{align}
Set
\begin{equation}\label{WE_def}
	W_{\mathcal E}(k):=V_{\mathcal E}(k)-I.
\end{equation}
By the construction of the global and local parametrices, the jumps of \(\mathcal E\)
are small in the following sense. Away from the disks, the jump matrices on
\(\widetilde{\mathcal C}_{\pm1}\cup\widetilde{\mathcal C}_{\pm2}\) are exponentially
close to the identity. Hence, for some constant \(c>0\),
\begin{equation}\label{WE_exp_est}
	\|W_{\mathcal E}\|_{(L^1\cap L^2\cap L^\infty)
		((\widetilde{\mathcal C}_{\pm1}\cup
		\widetilde{\mathcal C}_{\pm2})\setminus\mathcal D)}
	=
	\mathcal O(e^{-ct}),
	\qquad t\to+\infty.
\end{equation}
On the boundaries of the endpoint disks, the Bessel matching conditions give, 
\begin{equation}\label{WE_airy_est}
	\|W_{\mathcal E}\|_{L^\infty(
		\partial D_\epsilon(\pm i\eta_1)
		\cup
		\partial D_\epsilon(\pm i\eta_2))}
	=
	\mathcal O(t^{-1}),
	\qquad t\to+\infty.
\end{equation}
On the boundaries of the disks centered at \(\pm i\beta_1\), the parabolic-cylinder
matching conditions yield
\begin{equation}\label{WE_beta_boundary_est}
	\|W_{\mathcal E}\|_{L^\infty(\partial D_\epsilon(\pm i\beta_1))}
	=
	\mathcal O(t^{-1/2}),
	\qquad t\to+\infty.
\end{equation}
Moreover, on the local cross \(\mathcal X\cup \hat{\mathcal X} \), the mismatch between the exact jump and
the parabolic-cylinder model satisfies
\begin{equation}\label{WE_cross_est}
	\begin{aligned}
		\|W_{\mathcal E}\|_{L^1(\mathcal X \cup \hat{\mathcal X})}
		&=\mathcal O(t^{-1}\ln t),\\
		\|W_{\mathcal E}\|_{L^2(\mathcal X \cup \hat{\mathcal X})}
		&=\mathcal O(t^{-3/4}\ln t),\\
		\|W_{\mathcal E}\|_{L^\infty(\mathcal X \cup \hat{\mathcal X})}
		&=\mathcal O(t^{-1/2}\ln t),
	\end{aligned}
	\qquad t\to+\infty.
\end{equation}
Combining \eqref{WE_exp_est}--\eqref{WE_cross_est}, we obtain
\begin{equation}\label{WE_total_est}
	\|W_{\mathcal E}\|_{(L^1\cap L^2)(\Sigma_{\mathcal E})}
	=
	\mathcal O(t^{-1/2}),
	\qquad
	\|W_{\mathcal E}\|_{L^\infty(\Sigma_{\mathcal E})}
	=
	\mathcal O(t^{-1/2}\ln t).
\end{equation}

Let \(C_{\mathcal E}\) be the Cauchy operator on \(\Sigma_{\mathcal E}\),
\[
(C_{\mathcal E}f)(k)
=
\frac{1}{2\pi i}
\int_{\Sigma_{\mathcal E}}\frac{f(s)}{s-k}\,ds,
\qquad k\in\mathbb C\setminus\Sigma_{\mathcal E},
\]
and let \(C_{\mathcal E,\pm}\) denote its non-tangential boundary values. Define
\[
C_{W_{\mathcal E}}f
:=
C_{\mathcal E,-}\bigl(fW_{\mathcal E}\bigr),
\qquad f\in L^2(\Sigma_{\mathcal E}).
\]
Then
\begin{equation}\label{CWE_bound}
	\|C_{W_{\mathcal E}}\|_{\mathcal B(L^2(\Sigma_{\mathcal E}))}
	\leq
	C\|W_{\mathcal E}\|_{L^\infty(\Sigma_{\mathcal E})}
	=
	\mathcal O(t^{-1/2}\ln t),
	\qquad t\to+\infty.
\end{equation}
Therefore, for all sufficiently large \(t\), the operator
\(I-C_{W_{\mathcal E}}\) is invertible on \(L^2(\Sigma_{\mathcal E})\). Let
\(\mu_{\mathcal E}\in (1,1)+L^2(\Sigma_{\mathcal E})\) be defined by
\begin{equation}\label{muE_def}
	\mu_{\mathcal E}
	=
	(1,1)+
	(I-C_{W_{\mathcal E}})^{-1}
	C_{W_{\mathcal E}}(1,1).
\end{equation}
Equivalently,
\[
\mu_{\mathcal E}
=
(1,1)+C_{W_{\mathcal E}}\mu_{\mathcal E}.
\]
The Neumann-series estimate gives
\begin{equation}\label{muE_est}
	\|\mu_{\mathcal E}-(1,1)\|_{L^2(\Sigma_{\mathcal E})}
	=
	\mathcal O(t^{-1/2}),
	\qquad t\to+\infty.
\end{equation}
Consequently, the error RH problem \eqref{E_RHP} has a unique solution for all
sufficiently large \(t\), and this solution is represented by
\begin{equation}\label{E_solution}
	\mathcal E(k)
	=
	(1,1)
	+
	\frac{1}{2\pi i}
	\int_{\Sigma_{\mathcal E}}
	\frac{\mu_{\mathcal E}(s)W_{\mathcal E}(s)}{s-k}\,ds.
\end{equation}
In particular, as $k\to \frac i2$, $\mathcal E(k)$ has the asymptotic expansion,
\begin{equation}\label{E_small_est}
	\mathcal E(k)=\mathcal E(\frac i2)+\mathcal E_1(y,t) (k-\frac i2)+\mathcal O((k-\frac i2)^2),
\end{equation}
where $\mathcal E_1(y,t)=\frac{1}{2\pi i}
\int_{\Sigma_{\mathcal E}}
\frac{\mu_{\mathcal E}(s)W_{\mathcal E}(s)}{(s-\frac i2)^2}\,ds$.
We next estimate the size of the two coefficients in
\eqref{E_small_est}. 
Put
\[
\Sigma_{\mathcal E}^{\prime}
:=
\bigl(\widetilde{\mathcal C}_{\pm1}\cup
\widetilde{\mathcal C}_{\pm2}\bigr)\setminus\mathcal D .
\]
On \(\Sigma_{\mathcal E}^{\prime}\), the jumps are exponentially close to the
identity. Therefore, by \eqref{WE_exp_est} and \eqref{muE_est},
\begin{equation}\label{E_i2_exp_part}
	\int_{\Sigma_{\mathcal E}^{\prime}}
	\frac{\mu_{\mathcal E}(s)W_{\mathcal E}(s)}
	{\left(s-\frac{i}{2}\right)^j}\,ds
	=
	\mathcal O(e^{-ct}),
	\qquad j=1,2,\ \ t\to\infty.
\end{equation}
Similarly, the matching estimates on the Bessel disks give,
\begin{equation}\label{E_i2_airy_part}
	\int_{\partial D_\epsilon(\pm i\eta_1)\cup
		\partial D_\epsilon(\pm i\eta_2)}
	\frac{\mu_{\mathcal E}(s)W_{\mathcal E}(s)}
	{\left(s-\frac{i}{2}\right)^j}\,ds
	=
	\mathcal O(t^{-1}),
	\qquad j=1,2,\ \ t\to\infty.
\end{equation}
On the local cross \(\mathcal X \cup \hat{\mathcal X}\), using \eqref{WE_cross_est} and
\eqref{muE_est}, we obtain
\begin{equation}\label{E_i2_cross_part}
	\int_{\mathcal X\cup \hat{\mathcal X}}
	\frac{\mu_{\mathcal E}(s)W_{\mathcal E}(s)}
	{\left(s-\frac{i}{2}\right)^j}\,ds
	=
	\mathcal O(t^{-1}\ln t),
	\qquad j=1,2,\ \ t\to\infty.
\end{equation}
For convenience, set
\begin{equation}\label{Aibeta-def}
	\mathcal A_{i\beta_1}(x,t)
	:=
	\frac{
		Y_{i\beta_1}(x,t,i\beta_1)m_1^X
		Y_{i\beta_1}(x,t,i\beta_1)^{-1}}
	{\psi_{i\beta_1}(i\beta_1)}.
\end{equation}
By \eqref{eq:mmod-mmu-integral}, we know the contribution from $\partial D_\epsilon( i\beta_1)$ is
\begin{align*}
	&
	\frac{1}{2\pi i}
	\int_{\partial D_\epsilon( i\beta_1)}
	\frac{\mu_{\mathcal E}(s)W_{\mathcal E}(s)}{(s-\frac i2)^j}\d s\\
	=&
	\frac{1}{2\pi i}
	\int_{\partial D_\epsilon( i\beta_1)}
	\frac{(1,1)W_{\mathcal E}(s)}{(s-\frac i2)^j}\d s
	+
	\frac{1}{2\pi i}
	\int_{\partial D_\epsilon( i\beta_1)}
	\frac{(\mu_{\mathcal E}(s)-(1,1))W_{\mathcal E}(s)}{(s-\frac i2)^j}\d s\\
	=&\frac{
		(1,1)Y_{i\beta_1}(x,t,i\beta_1)m_1^X
		Y_{i\beta_1}(x,t,i\beta_1)^{-1}}
	{\left(i\beta_1-\frac i2\right)^j\psi_{i\beta_1}(i\beta_1)}t^{-1/2}+\mathcal{O}(\|\mu_{\mathcal E}(s)-(1,1)\|_{L^2(\partial D_{\epsilon}(i\beta_1))}\|W_{\mathcal E}(s)\|_{L^2(\partial D_{\epsilon}(i\beta_1))})\\
	=&(1,1)\frac{\mathcal A_{i\beta_1}(x,t)}{\left(i\beta_1-\frac i2\right)^j}t^{-1/2}+\mathcal{O}(t^{-1}),\ \ j=1,2,\ \ t\to\infty.
\end{align*}
By the symmetry \eqref{sym}, we also have
\begin{align*}
	&
	\frac{1}{2\pi i}
	\int_{\partial D_\epsilon( -i\beta_1)}
	\frac{\mu_{\mathcal E}(s)W_{\mathcal E}(s)}
	{\left(s-\frac i2\right)^j}\,ds \\
	=&
	-
	(1,1)
	\frac{
		\sigma_1\mathcal A_{i\beta_1}(x,t)\sigma_1}
	{\left(-i\beta_1-\frac i2\right)^j}
	t^{-1/2}
	+\mathcal O(t^{-1}),
	\qquad j=1,2,\quad t\to+\infty .
\end{align*}
Hence, we can rewrite \eqref{E_small_est} as
\begin{align}
	\mathcal{E}(k)=(1,1)+\mathcal E^{(1/2)}\left(\frac i2\right)t^{-1/2}+\mathcal E_1^{(1/2)}(y,t)t^{-1/2}(k-i/2)+\mathcal{O}(t^{-1}\ln t)+\mathcal{O}((k-i/2)^2),\ t\to\infty,
\end{align}
where 
\begin{align}
	&\mathcal E^{(1/2)}\left(\frac i2\right)=(1,1)\frac{\mathcal A_{i\beta_1}(x,t)}{\left(i\beta_1-\frac i2\right)}-(1,1)
	\frac{
		\sigma_1\mathcal A_{i\beta_1}(x,t)\sigma_1}
	{\left(-i\beta_1-\frac i2\right)},\\
	&\mathcal E_1^{(1/2)}(y,t)=(1,1)\frac{\mathcal A_{i\beta_1}(x,t)}{\left(i\beta_1-\frac i2\right)^2}-(1,1)
	\frac{
		\sigma_1\mathcal A_{i\beta_1}(x,t)\sigma_1}
	{\left(-i\beta_1-\frac i2\right)^2}.
\end{align}
Finally, trace back all the transformation we perform, $X(k)$ has the asymptotic expansion for $k\to\frac i2$,
\begin{align}\label{X-i2-expansion-explicit}
	X(k)
	=&\left[
	(1,1)
	+\frac{\mathcal E^{(1/2)}(\frac i2)}{\sqrt t}
	+\frac{\mathcal E_1^{(1/2)}(y,t)}{\sqrt t}
	\left(k-\frac i2\right)
	+\mathcal O\left(\frac{\ln t}{t}\right)
	+\mathcal O\left(\left(k-\frac i2\right)^2\right)
	\right]\nonumber\\
	&\times
	P^\infty(k)\tilde f(k)^{-\sigma_3}
	e^{-i(t\hat g(k)-\theta(k))\sigma_3}.
\end{align}
Write
\[
\mathcal E^{(1/2)}\left(\frac i2\right)=(e_1,e_2),
\qquad
\mathcal E_1^{(1/2)}(y,t)=(\widetilde e_1,\widetilde e_2).
\]
\begin{theorem}\label{thm:u-asymp-i2-explicit}
Assume that \(\eta_1\) and \(\eta_2\) fall into Case I. As \(t\to+\infty\), in the sector
\[
\xi\in(\hat\xi,0)\cup(0,\xi_*),
\]
or equivalently, in terms of the physical variable \(x\),
\[
\frac{x}{\omega t}\in\bigl(\Phi(\hat\xi),\Phi(0)\bigr)\cup\bigl(\Phi(0),\Phi(\xi_*)\bigr),
\]
where $\Phi$ is defined in \eqref{Phi},
the full CH soliton gas admits the following asymptotic behavior:
	\begin{align}
		u(x,t)
		=&
		\frac{4\omega(\eta_2^2-\eta_1^2)}
		{(1-4\eta_2^2)(1-4\eta_1^2)}
		-
		\frac{\eta_2\omega}
		{8K(\hat m)R(\frac i2)}
		\frac{\tilde F'\left(A(\frac i2)\right)}
		{\tilde F\left(A(\frac i2)\right)}
		\nonumber\\
		&+
		\frac{\omega}{2i\sqrt t}
		\left\{
		\partial_k
		\left[
		e_1+e_2
		+
		\frac{e_1-e_2}{2ik}
		\left(
		\frac{T^\infty_{1y}(k)}{T_1^\infty(k)}
		+
		\frac{T^\infty_{2y}(k)}{T_2^\infty(k)}
		\right)
		\right]_{k=\frac i2}
		\right.
		\nonumber\\
		&\qquad\qquad\left.
		+
		\left[
		\widetilde e_1+\widetilde e_2
		+
		\frac{\widetilde e_1-\widetilde e_2}{2ik}
		\left(
		\frac{T^\infty_{1y}(k)}{T_1^\infty(k)}
		+
		\frac{T^\infty_{2y}(k)}{T_2^\infty(k)}
		\right)
		\right]_{k=\frac i2}
		\right\}
		+
		\mathcal O\left(\frac{\ln t}{t}\right),
		\label{u-asymp-explicit}
	\end{align}
	and
		\begin{align}
		x
		&=
		y
		-2i\left(t\hat g-\theta\right)\!\left(\frac{i}{2}\right)
		+\log \tilde f^{-2}\!\left(\frac{i}{2}\right)
		+\log
		\tilde F(A(\frac i2))
		\nonumber\\
		&+
		\frac{e_1-e_2}{\sqrt t}
		\left[
		\frac{p(k)}{ik}
		+
		\frac{1}{2ik}
		\left(
		\frac{T^\infty_{1y}(k)}{T_1^\infty(k)}
		-
		\frac{T^\infty_{2y}(k)}{T_2^\infty(k)}
		\right)
		\right]_{k=\frac i2}
		+
		\mathcal O\left(\frac{\ln t}{t}\right),
		\label{xy-asymp-explicit}
	\end{align}
	\begin{align}
		\sqrt{\frac{m}{\omega}}
		&=
		\sqrt{\frac{1-4\eta_1^2}{1-4\eta_2^2}}\,
		\frac{\vartheta_3^2(0;2\tau)}
		{\vartheta_3^2\!\left(\dfrac{t\hat\Omega+\tilde\Delta}{2\pi};2\tau\right)}
		\tilde F(A(\frac i2))
		\nonumber\\
		&+
		\frac{T_1^\infty(\frac i2)T_2^\infty(\frac i2)}{\sqrt t}
		\left[
		e_1+e_2
		+
		\frac{e_1-e_2}{2ik}
		\left(
		\frac{T^\infty_{1y}(k)}{T_1^\infty(k)}
		+
		\frac{T^\infty_{2y}(k)}{T_2^\infty(k)}
		\right)
		\right]_{k=\frac i2}
		+
		\mathcal O\left(\frac{\ln t}{t}\right),
		\label{sqrt-m-asymp-explicit}
	\end{align}
	where $\tilde F(U)$ is given by 
	\begin{equation}\label{F2-def}
		\begin{aligned}
			\tilde F(U)
			&=
			\frac{
				\vartheta_3\!\left(
				2U+\dfrac{t\hat\Omega+\tilde\Delta}{2\pi}-\dfrac12;2\tau
				\right)}
			{
				\vartheta_3\!\left(
				2U-\dfrac12;2\tau
				\right)}
			\frac{
				\vartheta_3\!\left(
				-2U+\dfrac{t\hat\Omega+\tilde \Delta}{2\pi}-\dfrac12;2\tau
				\right)}
			{
				\vartheta_3\!\left(
				-2U-\dfrac12;2\tau
				\right)} .
		\end{aligned}
		\end{equation}
\end{theorem}

\begin{proof}
	From the definition of \(P^\infty(k)\), direct addition of its two rows gives
	\[
	(1,1)P^\infty(k)
	=
	(T_1^\infty(k),T_2^\infty(k)).
	\]
	Moreover, if
	\[
	\mathcal E^{(1/2)}\left(\frac i2\right)=(e_1,e_2),
	\]
	then a direct multiplication by \(P^\infty(k)\) gives
	\begin{align}
		\frac{
			\bigl[\mathcal E^{(1/2)}(\frac i2)P^\infty(k)\bigr]_1}
		{T_1^\infty(k)}
		&=
		\frac12
		\left[
		e_1+e_2
		+
		(e_1-e_2)\frac{p(k)}{ik}
		+
		\frac{e_1-e_2}{ik}
		\frac{T^\infty_{1y}(k)}{T_1^\infty(k)}
		\right],
		\label{E-P-ratio-1}\\
		\frac{
			\bigl[\mathcal E^{(1/2)}(\frac i2)P^\infty(k)\bigr]_2}
		{T_2^\infty(k)}
		&=
		\frac12
		\left[
		e_1+e_2
		-
		(e_1-e_2)\frac{p(k)}{ik}
		+
		\frac{e_1-e_2}{ik}
		\frac{T^\infty_{2y}(k)}{T_2^\infty(k)}
		\right].
		\label{E-P-ratio-2}
	\end{align}
	Multiplication by the diagonal factor
	\[
	\tilde f(k)^{-\sigma_3}
	e^{-i(t\hat g(k)-\theta(k))\sigma_3}
	\]
	does not affect the ratios in \eqref{E-P-ratio-1}--\eqref{E-P-ratio-2}, except for
	the leading ratio \(X_1(\frac i2)/X_2(\frac i2)\), where it contributes the factor
	\[
		\tilde f(\frac i2)^{-2}
		e^{-2i(t\hat g(\frac i2)-\theta(\frac i2))}.
	\]
	Hence the recovery relation
	\[
	\frac{X_1(\frac i2)}{X_2(\frac i2)}=e^{x-y}
	\]
	together with \eqref{E-P-ratio-1}--\eqref{E-P-ratio-2} gives
	\eqref{xy-asymp-explicit}.
	
	Next, since the two diagonal entries of
	\(\tilde f^{-\sigma_3}e^{-i(t\hat g-\theta)\sigma_3}\) multiply to one, the product
	\(X_1(\frac i2)X_2(\frac i2)\) has leading term
	\[
	T_1^\infty\left(\frac i2\right)
	T_2^\infty\left(\frac i2\right).
	\]
	Using the sum of \eqref{E-P-ratio-1} and \eqref{E-P-ratio-2}, we obtain
	\eqref{sqrt-m-asymp-explicit}.
	
	Finally, differentiating
	\[
	X_1(k)X_2(k)
	=
	\sqrt{\frac{m}{\omega}}
	\left[
	1+\frac{2i}{\omega}u(x,t)
	\left(k-\frac i2\right)
	+\mathcal O\left(\left(k-\frac i2\right)^2\right)
	\right]
	\]
	at \(k=\frac i2\) gives
	\[
	u(x,t)
	=
	\frac{\omega}{2i}
	\left.
	\partial_k\log\bigl(X_1(k)X_2(k)\bigr)
	\right|_{k=\frac i2}.
	\]
	The diagonal factor again cancels in the product \(X_1X_2\). Therefore the leading
	term is
	\[
	\frac{\omega}{2i}
	\left[
	\partial_k\log
	\bigl(T_1^\infty(k)T_2^\infty(k)\bigr)
	\right]_{k=\frac i2}.
	\]
	For the \(t^{-1/2}\) correction, we differentiate the sum of the two ratios in
	\eqref{E-P-ratio-1}--\eqref{E-P-ratio-2}. In addition, the term
	\(\mathcal E_1^{(1/2)}(y,t)(k-\frac i2)\) in \eqref{X-i2-expansion-explicit}
	contributes the same expression with \((e_1,e_2)\) replaced by
	\((\widetilde e_1,\widetilde e_2)\). This yields \eqref{u-asymp-explicit}.
	The proof is complete.
\end{proof}

\section*{Data availability}
No data was used for the research described in the article.

\section*{Acknowledgments}
This work is supported by the National Natural Science Foundation of China (Grant Nos. 12471234, 12401320, 12471240) and Science Foundation of Henan Academy of Sciences (Grant No. 20252319002).

\end{document}